\documentclass[a4work]{amsart}

\usepackage[leqno]{amsmath}
\usepackage{amsthm}
\usepackage{amssymb}
\usepackage[]{graphicx}
\usepackage{adjustbox}
\usepackage{calc} % simple arithmetics in latex
\usepackage{txfonts}
\usepackage{enumitem}
\usepackage{ifpdf} % define \ifpdf
\usepackage{ifthen}
\usepackage{braket}
\usepackage{xcolor}
\usepackage{todonotes}
\usepackage{algorithm}
\usepackage{algpseudocode}
\usepackage{tikz}
\usetikzlibrary{shapes.geometric}

\renewcommand{\Comment}[2][.6\linewidth]{%
  \leavevmode\hfill\makebox[#1][l]{//~#2}}

\usepackage[all]{xy}
\definecolor{darkblue}{RGB}{0,0,139}
\definecolor{darkgreen}{RGB}{0,139,0}
\definecolor{darkred}{RGB}{139,0,0}
\definecolor{sprgreen}{RGB}{0,250,154}
\definecolor{viol}{RGB}{199,21,133}
\definecolor{royblue}{RGB}{65,105,225}

\definecolor{navy}{RGB}{102,153,255}
\definecolor{tuerkis}{RGB}{51,153,204}

\theoremstyle{plain}
\newtheorem{theorem}{Theorem}[section]
\newtheorem{lemma}[theorem]{Lemma}
\newtheorem{proposition}[theorem]{Proposition}

\theoremstyle{definition}

\newtheorem{remark}[theorem]{Remark}

\newcommand{\R}{\mathbb{R}}
\newcommand{\N}{\mathbb{N}}

\newcommand{\f}[2]{\frac{#1}{#2}}

\newcommand{\half}{\f{1}{2}}

\renewcommand{\d}{\operatorname{d\!}} 
\newcommand{\isdef}{\mathrel{\mathrel{\mathop:}=}}

\newcommand{\bary}{\overline{y}}
\newcommand{\barY}{\overline{Y}}
\newcommand{\barD}{\overline{D}}
\newcommand{\barz}{\overline{z}}

\newcommand{\abs}[1]{\left\lvert#1\right\rvert} % absolute value
\newcommand{\ip}[2]{\left\langle #1\,,#2\right\rangle} % inner product (< , >)
\DeclareMathOperator{\diag}{diag}

\begin{document}

\author[C. Bayer]{Christian Bayer}
\address{Weierstrass Institute, Mohrenstrasse 39, 10117 Berlin, Germany}
\email{christian.bayer@wias-berlin.de}

\author[M. Siebenmorgen]{Markus Siebenmorgen}
\address{Institute for Numerical Simulation, University of Bonn, Wegelerstr. 6, 53115 Bonn, Germany}
\email{siebenmo@ins.uni-bonn.de}

\author[R. Tempone]{Raul Tempone}
\address{CEMSE, King Abdullah University of Science and Technology (KAUST), Thuwal 23955-6900, Saudi Arabia}
\email{raul.tempone@kaust.edu.sa}
\thanks{R.~Tempone is a member of the KAUST Strategic Research Initiative,
  Center for Uncertainty Quantification in Computational Sciences and
  Engineering.}

\title{Smoothing the payoff for efficient computation of Basket option
  prices} 

\thanks{We are grateful to Juho H\"{a}pp\"{o}l\"{a} for pointing out the
  immediate applicability of our method to Variance-Gamma processes.}



\begin{abstract}
  We consider the problem of pricing basket options in a multivariate
  Black-Scholes or Variance-Gamma model. From a numerical point of view,
  pricing such options corresponds to moderate and high-dimensional numerical
  integration problems with non-smooth integrands. Due to this lack of
  regularity, higher order numerical integration techniques may not be
  directly available, requiring the use of methods like Monte Carlo
  specifically designed to work for non-regular problems. We propose to use
  the inherent smoothing property of the density of the underlying in the
  above models to mollify the payoff function by means of an exact conditional
  expectation. The resulting conditional expectation is unbiased and yields a
  smooth integrand, which is amenable to the efficient use of adaptive
  sparse-grid cubature. Numerical examples indicate that the high-order method
  may perform orders of magnitude faster than Monte Carlo or Quasi Monte Carlo
  methods in dimensions up to $35$.
\end{abstract}

\maketitle
  
\section{Introduction}
\label{sec:introduction}

In quantitative finance, the price of an option on an underlying $S$ can
typically---disregarding discounting---be expressed as $E[f(S)]$ for some
(payoff) function $f$ on $S$ and the expectation operator $E$ induced by the
appropriate pricing measure. Hence, option pricing is an integration
problem. The integration problem is usually challenging due to a combination
of two complications:
\begin{itemize}
\item $S$ often takes values in a high-dimensional space. The reason for the
  high dimensionality may be time discretization of a stochastic differential
  equation, path dependence of the option (i.e., $S$ is actually a path of an
  asset price, not the value at a specific time), a large number of underlying
  assets, or others.
\item the payoff function $f$ is typically not smooth.
\end{itemize}
In this work, we focus on the problem of pricing basket options in models,
where the distribution of the underlying is explicitly given to
us. Specifically, we consider multivariate Black-Scholes and Variance-Gamma
models, i.e., models, for which no time discretization is required. We
consider a basket option on a $d$-dimensional underlying asset $S_T =
\left(S^1_T, \ldots, S^d_T \right)$ with payoff function
\begin{equation*}
  f(S_T) = \left( \sum_{i=1}^d w_i S^i_T - K \right)^+
\end{equation*}
for some positive weights $w_1, \ldots, w_d$, a maturity $T$ and a strike
price $K$. Observe in passing that one could also allow some weights to be
negative, an option type known as ``spread option''. Note that in addition,
(discrete) Asian options also fall under this framework.

Even in the standard Black-Scholes framework, closed-form expressions for
basket option prices are not available, since sums of log-normal random
variables are generally not log-normally distributed. Some explicit
approximation formulas are based on approximate distributional identities of
sums of log-normal random variables; for instance, see
\cite{Duf04,KdKKM04}. In addition, Laplace's method, possibly coupled with
heat kernel expansions when the distribution of the factors $S^i_T$ are given
only as solutions of stochastic differential equations, has been shown to
yield highly exact results even in high dimensions
(\cite{BL14,BFL15,BL15}). In this work, however, we aim to solve the problem
at hand using generic numerical integration techniques, which remain available
beyond the restrictions of the previous methods.

Efficient numerical integration algorithms are even available in high
dimensions, but they usually require smoothness of the integrand. Hence, they
are a priori not applicable in many option pricing problems. We will
specifically focus on (adaptive) sparse-grid methods, see
e.g.~\cite{BG04,GG03}.

Another efficient numerical integration technique is Quasi Monte Carlo
(QMC). Formally, QMC methods also rely on smoothness of the integrand to
retain first order convergence (up to multiplicative logarithmic terms);
however, QMC methods typically work very well for integration problems in
quantitative finance, even when the theoretically required regularity of the
integrand is not satisfied (see \cite{LE09} for an overview). In a series of
works, Griebel, {Kuo} and Sloan \cite{GKS10,GKS14,GKS15} analyzed the
performance of QMC methods for typical option pricing problems based on the
ANOVA decomposition. In particular, they show that all terms of the ANOVA
decomposition are smooth except for the last one. In the context of barrier
options, Achtsis, Cools and Nuyens \cite{ACN13a,ACN13b} successfully applied
QMC using a conditional sampling strategy to fulfill the barrier
conditions. Moreover, they use a root-finding procedure to determine the
region where the payoff function of the option is positive. In other words,
this procedure, which {is similar to the one discussed}
in~\cite{G07,H11}, locates the non-smooth part of the payoff function. Note
that the boundary of the support of the payoff function may be quite
complicated in terms of the coordinates for the integration problem, an issue
that may limit the applicability of such an approach.

From a numerical analysis point of view, the most obvious solution to the
problem is to smoothen the integrand using standard mollifiers, and there is a
prominent history of successful application of mollification in quantitative
finance; for instance, see \cite{FLLLT99} in the context of computing
sensitivities of option prices. For many financial applications, there seems
to be a more attractive approach that avoids the balancing act between
providing the smoothness needed for the numerical integration algorithm and
introducing bias in the integrand. Indeed, we suggest using the smoothing
property of the distribution of the underlying itself for regularizing the
integrand. This technique is quite standard in a time-stepping setting, and we
indeed plan to explore {its} applicability in that context in the future.

In this work, however, the regularization will be achieved by integrating
against one factor of the multivariate geometric Brownian motion
first---conditioning on all the other factors. More specifically, we show in
Section~\ref{sec:smoothing-payoff} below that we can always decompose
\begin{equation*}
  \sum_{i=1}^d w_i S^i_T \overset{\mathcal{L}}{=} H e^{Y}
\end{equation*}
for two independent random variables $H$ and $Y$---{here,
  ``$\overset{\mathcal{L}}{=}$'' denotes equality in law}. For the precise, explicit
construction see Lemma~\ref{lem:BS-formula} together with
Lemma~\ref{lem:rank-reduction}. Here, the random variable $Y$ is normally
distributed. Therefore, by computing the conditional expectation given $H$,
the basket option valuation problem is reduced to an integration problem in
$H$ (corresponding to an integration in $\R^{d-1}$) with a payoff function
given in this case by the Black-Scholes formula, a smooth function. 

The idea of integrating out one factor first, thereby obtaining an ``option''
on the remaining factors with payoff function giving by the Black-Scholes
formula is not new in finance. For instance, Romano and Touzi~\cite{RT97} have
applied this idea in a theoretical study of stochastic volatility models as a
tool to show convexity of admissible prices; in this vein, also see the work
\cite{G06}. The above mentioned decomposition (allowing the use of this trick
in the basket option context), however, seems new. As conditional expectations
always reduce the variance of a random variable, this trick can also be useful
in a Monte Carlo setting as well. {In this sense, the method
  is similar to the one proposed in~\cite{ACN13a,ACN13b}, who also reduce the
  variance of a (Q)MC estimator of a barrier option prices by clever
  transformation of the integrand coupled with identification of the region of
positive payoff values in terms of the integration variables. The approach
presented here is different since we really focus on the smoothing aspect
(obtaining lower variance as a welcome by-product), whereas the previous
approach is really focused on the variance, obtaining a smoother integrand as
a by-product. And indeed, even if applied to the special case of basket
options, the method of~\cite{ACN13a,ACN13b} will give a different result.}

{
We note in passing that the dependence of the convergence rate of our methodology is problem dependent. Indeed, the convergence rate depends both on the effective smoothing that the conditional expectation step introduces and the effective dimension of the resulting $d-1$ dimensional integration problem.  As an initial step towards a quantitative understanding of this dependence, Section \ref{sec:smoothing-payoff}  includes a  lower bound estimate on the effect of the smoothing. The development of more precise estimates are out of the scope of this work. 
}

Note that the smoothing approach proposed in this work can be applied in a
more general manner, possibly in modified ways, including more complicated
models, where the asset price process can only be simulated by a time-stepping
procedure. {We come back to this idea in the Conclusions.}
% In that case, we may no longer obtain an explicit, exact formula
% for the smoothed payoff through the conditional expectation step. Still, a
% properly constructed numerical quadrature to the conditional expectation will
% still inherit the fast convergence rates.

\subsection*{Outline}
\label{sec:outline}

We start by describing the setting of the problem in more detail. In
Section~\ref{sec:remind-effic-multi} we recall two popular efficient numerical
integration techniques for high dimensions, namely (adaptive) sparse-grids and
QMC. Then, in Section~\ref{sec:smoothing-payoff}, we describe the smoothing of
the payoff in the multivariate Black-Scholes framework.  Confirming the
exploratory style of this work, we give two detailed numerical examples. In
Section~\ref{sec:mult-black-schol}, we present numerical results for the
multivariate Black-Scholes model, and in
Section~\ref{sec:numerical-example-2}, we consider a multivariate
Variance-Gamma model, indicating that the smoothing method proposed here is
applicable beyond the standard Black-Scholes regime. Afterwards, we present
some concluding remarks including an outlook on future research.

\subsection*{Setting}
\label{sec:setting}

We consider a European basket option in a Black-Scholes model. More
specifically, we assume that the interest rate $r = 0$ -- i.e., we are working
with forward prices. We consider $d \in \N$ assets with prices $S_t =
\left(S^1_t, \ldots, S^d_t \right)$, $t > 0$, with risk-neutral dynamics
\begin{equation}
  \label{eq:BS-SDE}
  dS^i_t = \sigma_i S^i_t dW^i_t, \quad i=1, \ldots, d, 
\end{equation}
for volatilities $\sigma_i > 0$, $i = 1, \ldots, d$, driven by a correlated
$d$-dimensional Brownian motion $W$ with
\begin{equation*}
  d\ip{W^i}{W^j}_t = \rho_{i,j} dt, \quad i,j = 1, \ldots, d.
\end{equation*}
Obviously, \eqref{eq:BS-SDE} has the explicit solution
\begin{equation}
  \label{eq:GBM}
  S^i_t = S^i_0 \exp\left( - \half \sigma^2_i t + \sigma_i W^i_t \right),
  \quad i=1, \ldots, d,\ t > 0.
\end{equation}
We note that the components of the random vector $S_t$ have log-normal
distributions and are correlated.

A basket option is an option on such a collection of assets. We assume a
standard call option with strike $K > 0$ and maturity $T > 0$ with price
\begin{equation}
  \label{eq:basket-option}
  C_{\mathcal{B}} \coloneqq E\left[ \left( \sum_{i=1}^d c_i S^i_T - K\right)^+
  \right]. 
\end{equation}
Let us next transform the pricing problem~\eqref{eq:basket-option} into a
slightly more abstract form. As already observed, the random vector $\left(c_1
S^1_T, \ldots, c_d S^d_T\right)$ can be represented as $\left(w_1 e^{X_1},
\ldots, w_d e^{X_d} \right)$ for scalars $w_1, \ldots, w_d$ and a zero-mean
Gaussian vector $X = (X_1, \ldots, X_d) \sim \mathcal{N}(0, \Sigma)$. Indeed,
we may choose
\begin{gather*}
  w_i = c_i S^i_0 e^{- \half \sigma_i^2 T}, \quad i=1, \ldots, d, \\
  \Sigma_{i,j} = \sigma_i \sigma_j \rho_{i,j} T, \quad i,j = 1, \ldots, d.
\end{gather*}
Therefore, we are left with the problem of computing
\begin{equation}
  \label{eq:call-gauss}
  E\left[ \left( \sum_{i=1}^d w_i e^{X_i} - K \right)^+ \right]
\end{equation}
for $X \sim \mathcal{N}(0, \Sigma)$ and $d\ge1$.

\begin{remark}
  Note that the problem of computing the price of a (discretely monitored)
  Asian option on a 1D Black-Scholes asset is of the
  form~\eqref{eq:call-gauss} as well, but with a different covariance matrix
  $\Sigma$.
\end{remark}

In Section~\ref{sec:numerical-example-2}, we will also consider a
Variance-Gamma model; see \cite{MCC98} for the univariate and \cite{LS15} for
the multivariate Variance-Gamma model. We first recall the univariate case: Let
\begin{equation}
  \label{eq:var-gamma-logprice}
  X_t \coloneqq \theta {\gamma_t} + \sigma W_{\gamma_t}
\end{equation}
for {a real parameter $\theta$ (allowing control of the
  skewness),} a standard Brownian motion $W$ and an independent $\Gamma$
process $\gamma_t$ with parameters $1$ and $\nu > 0$ (i.e., $\gamma$ is a
process with stationary, independent increments with $\gamma_{t+h} - \gamma_t$
$\Gamma$-distributed with mean $h$ and variance $\nu h$, for any $h > 0$,
$t > 0$). Additionally, we impose $\gamma_0 = 0$. Under the risk-neutral
measure with $r=0$ (for simplicity), we then consider the asset price process
\begin{equation}
  \label{eq:var-gamma-price}
  S_t = S_0 \exp\left( \omega t + X_t \right), \quad \omega = \f{\log(1 -
    \theta \nu - \sigma^2 \nu /2)}{\nu};
\end{equation}
see \cite[formula (22)]{MCC98}. {The above choice of ``drift''
$\omega$ ensures that $S$ is a martingale.} Notice that the process $X$ is a L\'{e}vy
process and can alternatively be described as the difference of two
independent $\Gamma$ processes.

Economically, the time change $\gamma$ is often interpreted as ``business'' or
``trading'' time. Hence, it makes sense to assume that different stocks are
subject to a single time change. A reasonable multivariate generalization of
the Variance-Gamma model (also adopted in \cite{LS15}) requires defining
log terms $X^i_t$ as in \eqref{eq:var-gamma-logprice} based on
\emph{correlated} Brownian motions $W^i_t$, parameters $\theta_i$, $\sigma_i$,
but a common $\Gamma$-process $\gamma_t$ (hence, with a fixed parameter
$\nu$). The stock price components $S^i_t$, $i=1, \ldots, d$, are then defined
according to \eqref{eq:var-gamma-price} based on $X^i_t$, $\theta_i$,
$\sigma_i$, but the common parameter $\nu$.

\section{{A brief overview of efficient multi-dimensional numerical integration}}
\label{sec:remind-effic-multi}

In this section, we give a brief review on efficient multidimensional
integration schemes, in particular the Monte Carlo quadrature, the QMC
quadrature and the adaptive sparse-grid quadrature. To this end, let us
consider a function \(f\colon \R^d\to \R\) and denote the \(d\)-dimensional
standard Gaussian density function by \(\phi_d\colon \R^d \to \R_+,\ x\mapsto
(2\pi)^{-d/2}\prod_{k=1}^d \exp(-x_k^2/2)\). As we will see later on, the
multi-dimensional integration problem that we are faced with is to find an
approximation to the integral
\begin{equation}\label{eq:multiint}
\int_{\R^d} f(x)\phi_d(x) \d x.
\end{equation}

\subsection{Monte Carlo and Quasi Monte Carlo quadrature}

The most widely used quadrature technique to tackle high-dimensional
integration problems is the Monte Carlo quadrature; for example, see
\cite{HH64}.  This quadrature draws \(N\in \N\) independently and identically
distributed samples \(\xi_i\in\R^d,\ i=1,\ldots,N\) with respect to the
\(d\)-dimensional standard normal distribution. Then, the unbiased Monte Carlo
estimator for the integral \eqref{eq:multiint} is given by
\begin{equation}\label{eq:MCest}
\int_{\R^d} f(x)\phi_d(x) \d x\approx\frac 1 N \sum_{i=1}^N f(\xi_i).
\end{equation}
The big advantage of this quadrature is that {the root mean 
square error } converges with a rate that is
independent of the dimensionality \(d\), but the convergence rate
\(\mathcal{O}(N^{-1/2})\) is rather low. Another advantage of this quadrature
is that it works under low regularity requirements on the integrand. To be
more precise, the variance of the integrand is a multiplicative constant in
the error estimate.

The QMC quadrature is of the same form \eqref{eq:MCest} as the MC quadrature,
but the sample points \(x_i\) are constructed or taken from a prescribed
sequence rather than chosen randomly.  There are several QMC sequences
available in the literature, see e.g.~\cite{Caf98,Nie92} {or
  \cite{dick2013high} for a recent review article}. Nevertheless, almost
all QMC sequences refer to integration over the unit cube \([0,1]^d\) with
respect to the Lebesgue measure and, hence, these points have to be mapped to
the domain of integration {\(\R^d\)} by the inverse normal distribution.  The
aim of a QMC sequence is to mirror with the first \(N\) sample points the
uniform distribution on the unit cube as accurately as possible. A measure of
the distance between the uniform distribution and the first \(N\) sample
points is then given by the discrepancy of these sample points (see
\cite{Nie92}). This is because the QMC integration error for functions with
bounded variation in the sense of Hardy and Krause can be estimated up to a
constant by the discrepancy of the integration points.  A QMC sequence is
called a low-discrepancy sequence if the discrepancy of the first \(N\) points
of this sequence is \(\mathcal{O}(N^{-1}\log(N)^{d})\).  Thus, low-discrepancy
sequences can improve the convergence of the Monte-Carlo quadrature. In our
numerical examples, we will use the QMC quadrature based on the
Sobol-sequence, cf.~\cite{Sob67}, which is a classic low-discrepancy sequence.

{It should be noted that modern applications of QMC methods,
  in particular in finance, have moved away from the classical, deterministic
  low discrepancy sequences mentioned above. Instead, \emph{randomized}
  sequences are usually applied, which provide both the speed of convergence
  of classical QMC and the simple yet accurate error control provided by MC
  methods. We refer once again to \cite{dick2013high} for a general overview
  and to \cite{LE09} for a specific review for financial applications. While
  extremely important in general, we completely ignore the issue of error
  control in this work, instead concentrating on ``raw performance''.}

\subsection{Adaptive sparse-grid quadrature}
\label{sec:sparsegridconstr}

The construction of a sparse-grid quadrature is based on a sequence of 1D
quadrature rules (cf.~\cite{BG04,Smo63}). Hence, we define for a function
\(f\colon\R\to \R\) quadrature rules
\begin{equation}
  \label{eq:interpp}
  \int_{\R} f(x) \phi_1(x) \d x \approx Q_j(f) = \sum_{i=1}^{N_j} w_{i}^{(j)}
  f\left(\eta_{i}^{(j)}\right),\quad N_j\in\mathbb{N},\quad j=0,1,\ldots 
\end{equation}
with suitable quadrature points and weights
\(\left\{\left(\eta_{i}^{(j)},w_{i}^{(j)}\right)\right\}_{i=1}^{N_j}\subset\R\times\R\).
Usually, the sequence of quadrature rules is increasing (i.e.,~\(N_0< N_1 <
\ldots\)) and the first quadrature rule uses only one quadrature point and
weight (i.e.,~\(N_0=1\)).  According to the sequence \(\{Q_j\}_j\), we
introduce the difference quadrature operator
\begin{equation}\label{eq:deltaop}
%==================================
\Delta_j\isdef Q_{j}-Q_{j-1},\quad\text{where}\quad 
Q_{-1}\isdef 0.
\end{equation} 

Assume that the sequence \(\{Q_jf\}_j\) converges, that is
\[
\int_{\R} f(x) \phi_1(x) \d x = \lim_{n\to \infty} Q_n f =
\lim_{n\to\infty}\sum_{j=0}^n \Delta_j f.
\]
This implies that the sequence \(\{|\Delta_j f|\}_j\) converges to zero and
hence the importance of the difference of the quadrature operators decays in
\(j\). Unfortunately, this decay is not necessarily monotonic, but it builds
the basic idea of adaptive sparse-grid constructions.

With the difference quadrature operators \(\Delta_j\) at hand, a generalized
sparse-grid quadrature for the integration problem \eqref{eq:multiint} is
defined by
\begin{equation}
\label{eq:gensg}
\int_{\R^d} f(x)\phi_d(x) \d x\approx \sum_{\alpha\in \mathcal{I}} 
\Delta_{\alpha} f\isdef
\sum_{\alpha\in \mathcal{I}} 
\Delta_{\alpha_1}\otimes \Delta_{\alpha_2}\otimes \cdots\otimes \Delta_{\alpha_d} f  
\end{equation}
for an admissible index set \(\mathcal{I}\subset \N_0^d\). Such an index set
\(\mathcal{I}\) is called admissible if it holds for \(j=1,\ldots,n\) and the
unit multi-index \(e_j\) that
\[
\alpha \in \mathcal{I} \Longrightarrow \alpha-e_j \in \mathcal{I} \quad\text{if \(\alpha_j>0\).}
\]
As can be seen from \eqref{eq:deltaop} and \eqref{eq:gensg}, a generalized
sparse-grid quadrature is uniquely determined by a sequence of univariate
quadrature rules \(\{Q_j\}_j\) and an admissible index set \(\mathcal{I}\).
The index set \(\mathcal{I}\) can be chosen a priori, for example as
\begin{equation}
\label{eq:classsg}
\mathcal{I} = \left\{\alpha\in \N_0^d: \sum_{i=1}^n \alpha_i \le q\right\}
\end{equation}
which corresponds to a total-degree sparse-grid on level \(q\). 

Another option is to {adaptively expand the index set} \(\mathcal{I}\).  In this case an
initial index set is selected, most often \(\mathcal{I} =
\{(0,\ldots,0)\}\).  Then, the integration error of the sparse-grid quadrature
with respect to \(\mathcal{I}\) is estimated by a local error estimator and,
afterwards, the indices with the largest local error estimator are
successively added to \(\mathcal{I}\) until a global error estimator
\(\eta\) has reached a certain tolerance. We denote the local error estimator
of an index \(\alpha\in\mathcal{I}\) by \(g_\alpha\) and for our purpose
we use the absolute value of the associated difference quadrature formula
(i.e., \(g_\alpha\isdef |\Delta_{\alpha} f|\)). Of course, we have to guarantee
during the algorithm that the admissible condition of \(\mathcal{I}\) is not
violated. A detailed description of this method is provided in
\cite{GG03}. We recall here the algorithm from \cite{GG03} and explain the
most important steps.

\begin{algorithm}[hbt]
\caption{Adaptive sparse-grid quadrature for a function $f$}
\label{alg:adap}
\begin{algorithmic}
\State $\alpha\gets  (0,\ldots,0)${\Comment{$\alpha$: Index associated with a local error estimator $g_{\alpha}$}}
\State $\mathcal{O}\gets  \emptyset${\Comment{$\mathcal{O}$: Old index set}}
\State $\mathcal{A}\gets  \alpha${\Comment{$\mathcal{A}$: Active index set}}
\State $y\gets \Delta_{\alpha} f${\Comment{$y$: Approximation to the value of the integral}}
\State $\eta\gets g_{\alpha}${\Comment{$\eta$: Global error estimator}}
\While{$(\eta>\text{TOL})$}

select $\alpha$ from $\mathcal{A}$ with largest $g_{\alpha}$
\State $\mathcal{A}\gets \mathcal{A}\setminus \alpha$
\State $\mathcal{O}\gets  \mathcal{O}\cup\alpha$
\State $\eta\gets \eta- g_{\alpha}$
\For {$(k =1,\ldots,d)$} 

\State $\beta \gets \alpha+e_k$
\If {$(\beta-e_q \in \mathcal{O}\ \text{for all $q=1,\ldots,d$})$}
\State $\mathcal{A}\gets \mathcal{A}\cup \beta$
\State $x \gets \Delta_{\beta}f$
\State $y \gets y+x$
\State $\eta\gets \eta+ g_{\beta}$
\EndIf
\EndFor
\EndWhile 
\State \Return y
\end{algorithmic}
\end{algorithm}

In Algorithm \ref{alg:adap}, the index set \(\mathcal{I}\) in \eqref{eq:gensg}
is partitioned into the \emph{old index set} \(\mathcal{O}\) and the
\emph{active index set} \(\mathcal{A}\). The active index set contains all
indices \(\alpha\) whose local error estimators \(g_{\alpha}\) actually
contribute to the global error estimator \(\eta\). Then, the element
\(\alpha\) of \(\mathcal{A}\) with the largest local error estimator is
removed from the active index set and entered into the old index set and the
children of \(\alpha\), i.e.~\(\alpha+e_j\), are successively added to the
active index set, as long as all their parents belong to the old index
set. The last step is necessary to guarantee the admissibility
condition. Then, the contribution of the new indices to the value of the
integral as well as the local and global error estimators is updated and the
procedure is repeated {until the global error estimator has 
reached a prescribed tolerance. To clarify the role of $\mathcal{A}$ and $\mathcal{O}$, 
we note that the following conditions are always satisfied during the algorithm 
\begin{enumerate}
\item $\alpha \in \mathcal{O} \Rightarrow (\alpha - e_q) \in \mathcal{O}$ for all $q = {1,...,d}$ with $\alpha_q>0$ 
which means that $\mathcal{O}$ is admissible, 
\item $\alpha \in \mathcal{A} \Rightarrow (\alpha - e_q) \in \mathcal{O}$ for all $q = {1,...,d}$ with $\alpha_q>0$, 
\item  $\alpha \in \mathcal{A} \Rightarrow (\alpha + e_q) \notin \mathcal{O}$ for all $q = {1,...,d}$. 
\end{enumerate}
In Figure \ref{fig:indsg}, the change in the current index set during two steps of the algorithm in $d=2$ 
dimensions is visualized. In the first step, both indices fulfill the admissibility check and are added to the 
active index set. In the second step, only the index $(3,1)$ is added to $\mathcal{A}$ while the admissibility check 
for the index $(2,2)$ fails. }

\begin{figure}[hbt]

\begin{minipage}{0.3\textwidth}
\begin{tikzpicture}[scale=.32,every node/.style={minimum size=1cm},on grid]
 \fill[white,fill opacity=0.9] (0,0) rectangle (7,4);
        \draw[step=1cm, black] (0,0) grid (7,4); %defining grids
%        \draw[step=1mm, red!50,thin] (3,1) grid (4,2);  %Nested Grid
        \draw[black,very thick] (0,0) rectangle (7,4);%marking borders
        \fill[tuerkis] (0.1,0.1) rectangle (0.9,0.9);
       	\fill[tuerkis] (1.1,0.1) rectangle (1.9,0.9);
	\fill[tuerkis] (2.1,0.1) rectangle (2.9,0.9);
	\fill[tuerkis] (3.1,0.1) rectangle (3.9,0.9);
	\fill[tuerkis] (4.1,0.1) rectangle (4.9,0.9);
	\fill[red] (5.1,0.1) rectangle (5.9,0.9);

        \fill[tuerkis] (0.1,1.1) rectangle (0.9,1.9);
        \fill[tuerkis] (1.1,1.1) rectangle (1.9,1.9);
        \fill[red] (2.1,1.1) rectangle (2.9,1.9);
        \fill[red] (0.1,2.1) rectangle (0.9,2.9);
\draw[thick,black] (0.5,2.5) circle (0.8);

 \draw[->,thick] (0,0) -- (7.5,0) node[right] {$\alpha_1$};
    	\draw[->,thick] (0,0) -- (0,4.5) node[left] {$\alpha_2$};
    	\foreach \x/\xtext in {0/0, 1/1, 2/2, 3/3, 4/4, 5/5,6/6}
    	\draw[shift={(\x+0.5,0)}] (0pt,2pt) -- (0pt,-2pt) node[below] {$\xtext$};
    	\foreach \y/\ytext in {0/0, 1/1, 2/2,3/3}
    	\draw[shift={(0,\y+0.5)}] (2pt,0pt) -- (-2pt,0pt) node[left] {$\ytext$};
%\node at (3.5cm,5.5cm) {Indices sparse-grid};
\draw[->,thick](9,2)--(11,2); 
\end{tikzpicture}
\end{minipage}
\hfill
\begin{minipage}{0.3\textwidth}
\begin{tikzpicture}[scale=.32,every node/.style={minimum size=1cm},on grid]
 \fill[white,fill opacity=0.9] (0,0) rectangle (7,4);
        \draw[step=1cm, black] (0,0) grid (7,4); %defining grids
%        \draw[step=1mm, red!50,thin] (3,1) grid (4,2);  %Nested Grid
        \draw[black,very thick] (0,0) rectangle (7,4);%marking borders
        \fill[tuerkis] (0.1,0.1) rectangle (0.9,0.9);
       	\fill[tuerkis] (1.1,0.1) rectangle (1.9,0.9);
	\fill[tuerkis] (2.1,0.1) rectangle (2.9,0.9);
	\fill[tuerkis] (3.1,0.1) rectangle (3.9,0.9);
	\fill[tuerkis] (4.1,0.1) rectangle (4.9,0.9);
	\fill[red] (5.1,0.1) rectangle (5.9,0.9);

        \fill[tuerkis] (0.1,1.1) rectangle (0.9,1.9);
        \fill[tuerkis] (1.1,1.1) rectangle (1.9,1.9);
        \fill[red] (2.1,1.1) rectangle (2.9,1.9);
        \fill[tuerkis] (0.1,2.1) rectangle (0.9,2.9);
 	\fill[red] (0.1,3.1) rectangle (0.9,3.9);
 	\fill[red] (1.1,2.1) rectangle (1.9,2.9);
 	\draw[thick,black] (2.5,1.5) circle (0.8);

 \draw[->,thick] (0,0) -- (7.5,0) node[right] {$\alpha_1$};
    	\draw[->,thick] (0,0) -- (0,4.5) node[left] {$\alpha_2$};
    	\foreach \x/\xtext in {0/0, 1/1, 2/2, 3/3, 4/4, 5/5,6/6}
    	\draw[shift={(\x+0.5,0)}] (0pt,2pt) -- (0pt,-2pt) node[below] {$\xtext$};
    	\foreach \y/\ytext in {0/0, 1/1, 2/2,3/3}
    	\draw[shift={(0,\y+0.5)}] (2pt,0pt) -- (-2pt,0pt) node[left] {$\ytext$};
%\node at (3.5cm,5.5cm) {Indices sparse-grid};
\draw[->,thick](9,2)--(11,2); 
\end{tikzpicture}
\end{minipage}
\hfill
\begin{minipage}{0.3\textwidth}
\begin{tikzpicture}[scale=.32,every node/.style={minimum size=1cm},on grid]
 \fill[white,fill opacity=0.9] (0,0) rectangle (7,4);
        \draw[step=1cm, black] (0,0) grid (7,4); %defining grids
%        \draw[step=1mm, red!50,thin] (3,1) grid (4,2);  %Nested Grid
        \draw[black,very thick] (0,0) rectangle (7,4);%marking borders
        \fill[tuerkis] (0.1,0.1) rectangle (0.9,0.9);
       	\fill[tuerkis] (1.1,0.1) rectangle (1.9,0.9);
	\fill[tuerkis] (2.1,0.1) rectangle (2.9,0.9);
	\fill[tuerkis] (3.1,0.1) rectangle (3.9,0.9);
	\fill[tuerkis] (4.1,0.1) rectangle (4.9,0.9);
	\fill[red] (5.1,0.1) rectangle (5.9,0.9);

        \fill[tuerkis] (0.1,1.1) rectangle (0.9,1.9);
        \fill[tuerkis] (1.1,1.1) rectangle (1.9,1.9);
        \fill[tuerkis] (2.1,1.1) rectangle (2.9,1.9);
        \fill[red] (3.1,1.1) rectangle (3.9,1.9);
        \fill[tuerkis] (0.1,2.1) rectangle (0.9,2.9);
 	\fill[red] (0.1,3.1) rectangle (0.9,3.9);
 	\fill[red] (1.1,2.1) rectangle (1.9,2.9);

 \draw[->,thick] (0,0) -- (7.5,0) node[right] {$\alpha_1$};
    	\draw[->,thick] (0,0) -- (0,4.5) node[left] {$\alpha_2$};
    	\foreach \x/\xtext in {0/0, 1/1, 2/2, 3/3, 4/4, 5/5,6/6}
    	\draw[shift={(\x+0.5,0)}] (0pt,2pt) -- (0pt,-2pt) node[below] {$\xtext$};
    	\foreach \y/\ytext in {0/0, 1/1, 2/2,3/3}
    	\draw[shift={(0,\y+0.5)}] (2pt,0pt) -- (-2pt,0pt) node[left] {$\ytext$};
%\node at (3.5cm,5.5cm) {Indices sparse-grid};
%\draw[->,thick](9,2)--(12,2); 
\end{tikzpicture}
\end{minipage}
\vspace*{2mm}
\begin{minipage}{0.4\textwidth}
\tikz{\path[draw=black,fill=tuerkis] (0,0) rectangle (.3cm,.3cm) node[right,yshift=-.15cm,xshift=.2cm] {Old Index set $\mathcal{O}$};}
\end{minipage}
\hfill
\begin{minipage}{0.4\textwidth}
\tikz{\path[draw=black,fill=red] (0,-.5cm) rectangle (.3cm,-.2cm) node[right,yshift=-.15cm,xshift=.2cm] {Active Index set $\mathcal{A}$};}
\end{minipage}
\caption{\label{fig:indsg}Two steps of the adaptive quadrature where $\alpha=(0,2)$ is the index 
with the largest \(g(\alpha)\) in the first step and $\alpha=(2,1)$ in the second step.}
\end{figure}
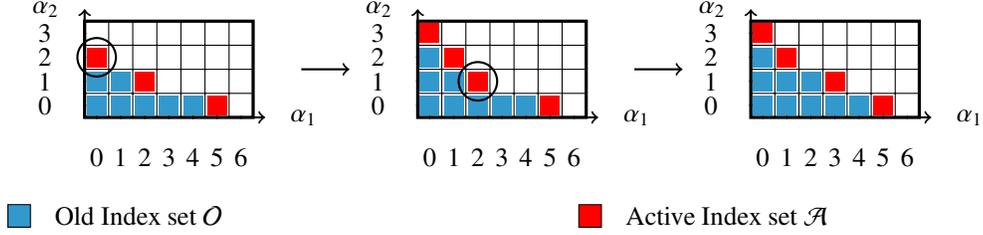

We will use Gau{\ss}-Hermite and Genz-Keister quadrature rules
(cf.~\cite{GK96}) as 1D sequences. Gau{\ss}-Hermite quadrature rules have the
highest degree of polynomial exactness for integrals as in \eqref{eq:interpp}
while Genz-Keister rules have the advantage that they are nested. More
precisely, the Genz-Keister rules are extensions of Gau{\ss}-Hermite
quadrature rules of relatively low degree. As the first extension of the
one-point Gau{\ss}-Hermite quadrature we use the three-point Gau{\ss}-Hermite
quadrature.  Further extensions do not coincide with any other
Gau{\ss}-Hermite quadrature rule.

At the end of this section, we visualize on the right-hand side of Figure
\ref{fig:grid} the 2D adaptive sparse-grid points which are used {for 
the approximation with TOL$=10^{-9}$} in our first
numerical example. On the left-hand side of Figure \ref{fig:grid}, we show
the associated adaptive index set.

\begin{figure}[hbt]

\begin{minipage}{0.4\textwidth}
\begin{tikzpicture}[scale=.5,every node/.style={minimum size=1cm},on grid]
 \fill[white,fill opacity=0.9] (0,0) rectangle (6,5);
        \draw[step=1cm, black] (0,0) grid (6,5); %defining grids
%        \draw[step=1mm, red!50,thin] (3,1) grid (4,2);  %Nested Grid
        \draw[black,very thick] (0,0) rectangle (6,5);%marking borders
        \fill[tuerkis] (0.1,0.1) rectangle (0.9,0.9);
       	\fill[tuerkis] (1.1,0.1) rectangle (1.9,0.9);
	\fill[tuerkis] (2.1,0.1) rectangle (2.9,0.9);
	\fill[tuerkis] (3.1,0.1) rectangle (3.9,0.9);
	\fill[tuerkis] (4.1,0.1) rectangle (4.9,0.9);
	%\fill[red] (5.1,0.1) rectangle (5.9,0.9);
 	\fill[tuerkis] (1.1,2.1) rectangle (1.9,2.9);
        \fill[tuerkis] (0.1,1.1) rectangle (0.9,1.9);
        \fill[tuerkis] (1.1,1.1) rectangle (1.9,1.9);
        \fill[tuerkis] (2.1,1.1) rectangle (2.9,1.9);
        \fill[tuerkis] (3.1,1.1) rectangle (3.9,1.9);
        \fill[tuerkis] (0.1,2.1) rectangle (0.9,2.9);
         \fill[tuerkis] (0.1,3.1) rectangle (0.9,3.9);
 	\fill[tuerkis] (2.1,2.1) rectangle (2.9,2.9);

 \draw[->,thick] (0,0) -- (6.5,0) node[right] {$\alpha_1$};
    	\draw[->,thick] (0,0) -- (0,5.5) node[left] {$\alpha_2$};
    	\foreach \x/\xtext in {0/0, 1/1, 2/2, 3/3, 4/4, 5/5}
    	\draw[shift={(\x+0.5,0)}] (0pt,2pt) -- (0pt,-2pt) node[below] {$\xtext$};
    	\foreach \y/\ytext in {0/0, 1/1, 2/2,3/3,4/4}
    	\draw[shift={(0,\y+0.5)}] (2pt,0pt) -- (-2pt,0pt) node[left] {$\ytext$};
%\node at (3.5cm,5.5cm) {Indices sparse-grid};
%\draw[->,thick](11,2)--(14,2); 
\end{tikzpicture}

\end{minipage}
\hfill
\begin{minipage}{0.5\textwidth}
\includegraphics[width=.8\textwidth]{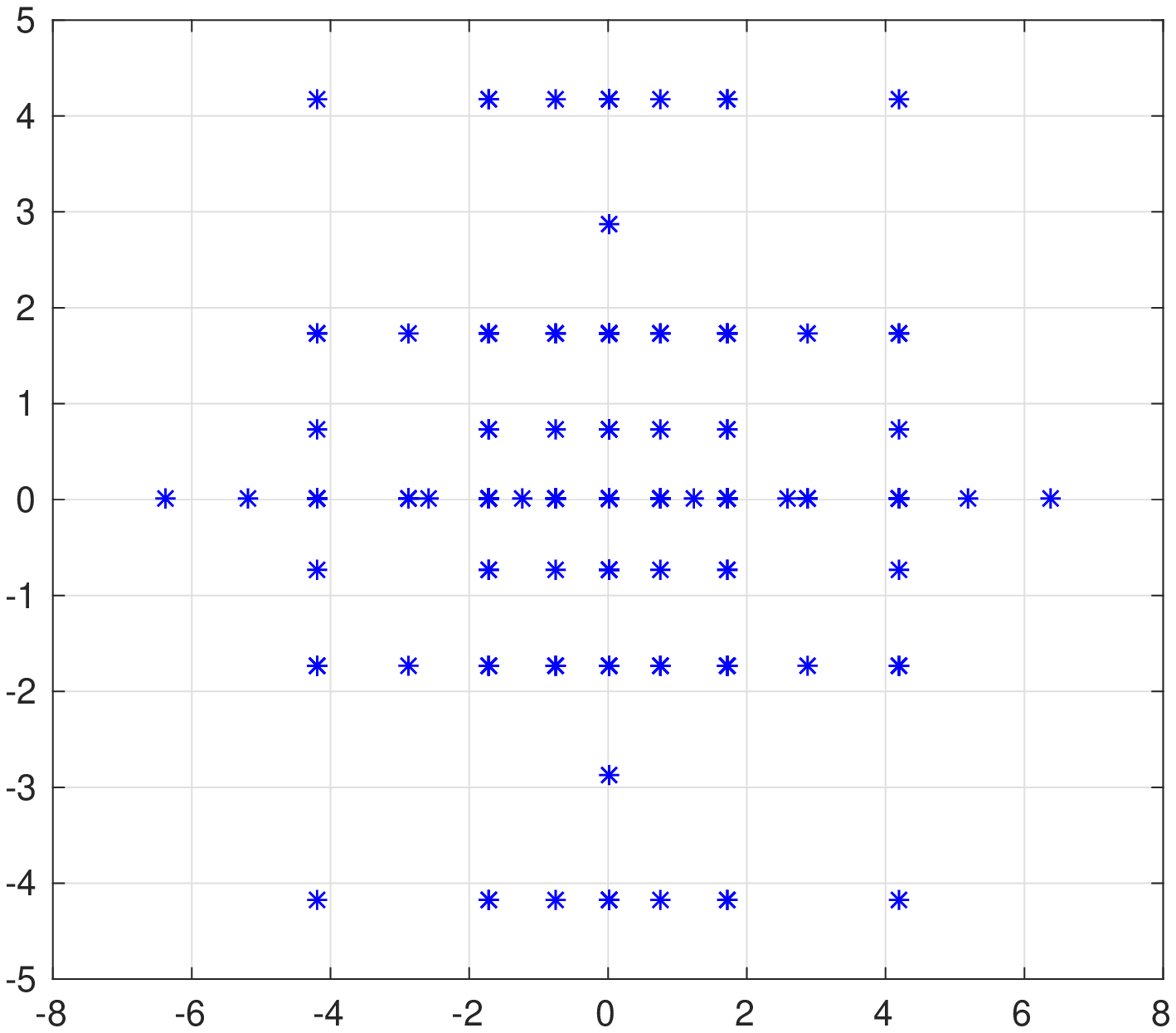}
\end{minipage}
\caption{\label{fig:grid} Index set \(\mathcal{I}\) of the sparse-grid on the
  left and the associated sparse-grid points, which are used in the first
  numerical example, on the right.}
\end{figure}

\begin{remark}
  A further alternative could be to use multidimensional cubature formulas;
 see for instance, \cite{Cools2003445}. In principle, high-order cubature
  formulas also require smoothness of the integrand, therefore we suspect that
  the approach presented here will also work well in the cubature context. 
  These methods are beyond the scope of the current paper.
\end{remark}

\section{Smoothing the payoff}
\label{sec:smoothing-payoff}

In this section, we will describe a simple technique for smoothing the
integrand in~\eqref{eq:call-gauss} which, at the same time,
\begin{itemize}
\item produces an analytic integrand;
\item does not introduce a {bias}  error;
\item reduces the variance of the {resulting} integrand.
\end{itemize}
For the following, we assume that the covariance matrix $\Sigma$ is
invertible (i.e., a positive definite symmetric matrix).

The general idea is that we want to integrate out one Gaussian factor
in~\eqref{eq:call-gauss}, conditioning on the remaining $d-1$
factors. Clearly, the outcome of such a procedure is a smooth function of the
remaining factors. However, generically there is no closed formula for this
function. The reason for this is that there is no closed formula for the
simple special case
\begin{equation*}
  E\left[\left( e^{\sigma_1 Z} + e^{\sigma_2 Z} - K \right)^+ \right]
\end{equation*}
for $Z \sim \mathcal{N}(0,1)$ and $\sigma_1 \neq \sigma_2$. Indeed,
$e^{\sigma_1 Z} + e^{\sigma_2 Z}$ has a log-normal distribution if and only if
$\sigma_1 = \sigma_2$. In this case, the above expression is given in terms of
the celebrated Black-Scholes formula, which will be reviewed below.

It turns out, that a {suitable} choice of factorization of the covariance matrix
of the Gaussian factors allows us to factor out one common, independent
log-normal term. This is a consequence of the

\begin{lemma}
  \label{lem:rank-reduction}
  Let $\Sigma$ be a symmetric, positive definite $d\times d$ matrix. Then
  there is {for each vector $v\in\R^d$} a diagonal matrix 
  $D = \diag\left( \lambda_1^2,  \lambda_d^2, \ldots, \lambda_d^2 \right)$ 
  and an invertible matrix $V \in\R^{d \times d}$ with the property that 
  $V_{i,1} \equiv v_i$, $i = 1, \ldots,d$, and
  \begin{equation*}
    \Sigma = V D V^{\top}.
  \end{equation*}
  Moreover, we may choose the remaining columns of $V$ such that 
  $\lambda_2^2\ge \ldots \ge \lambda_d^2 \ge 0$.
\end{lemma}
\begin{proof}
  From \cite[p.~126]{A02}, we know that for every \({ 0}\neq { s}\in
  \mathbb{R}^n\), the rank-1 modification
 \begin{equation}\label{eq:rank1}
 \tilde{{A}} = {A} - \frac{({As}) ({As})^{\top}}{{
     s}^{\top}  
 { As}}
 \end{equation}
 of a symmetric, positive definite matrix \({ A}\in \mathbb{R}^{d\times d}\)
 yields a symmetric and positive semi-definite matrix \(\tilde{A}\in
 \mathbb{R}^{d\times d}\) of rank \(d-1\). 
 Let us denote \({w} \coloneqq {\Sigma}^{-1}v\). Then it
 follows from \eqref{eq:rank1} that
 \[
 \tilde{\Sigma} = {\Sigma} - \frac{vv^{\top}}{
   v^{\top} {w}}
 \]
 is a symmetric and positive semi-definite matrix of rank \(d-1\).  Denote by
 \((\lambda_i^2,{ v}_i)\) for \(i=2,\ldots,d\) the \(d-1\) eigenpairs
 corresponding to the \(d-1\) positive eigenvalues of \(\tilde{\Sigma}\).
 Defining \(V=[{ v}_1,{v}_2,\ldots,{ v}_d]\) with \({ v}_1=v\) and
 \(D=\diag(\lambda_1^2, \lambda_2^2,\ldots, \lambda_d^2)\) with
 \(\lambda_1^2=(v^{\top} {w})^{-1}\) leads to the desired result.
\end{proof}

{For the following computations, we choose the vector $v$ from
  Lemma~\ref{lem:rank-reduction} as $v=\mathbf{1} \coloneqq [1,\ldots, 1]^\top$.}  In
the next step, we replace $X$ by $Y \coloneqq V^{-1}X \sim \mathcal{N}(0, D)$
and note that the components of $Y$ are independent. By substituting the
decomposition $X = VY$ into~\eqref{eq:call-gauss}, we obtain
\begin{align}
  C_{\mathcal{B}} &= E\left[ \left( \sum_{i=1}^d w_i e^{(VY)_i} - K \right)^+
                    \right] \nonumber\\
  &= E\left[ \left( \sum_{i=1}^d w_i \exp\left( Y_1 + \sum_{j=2}^d V_{i,j} Y_j
    \right) - K \right)^+ \right] \nonumber\\
  &= E\left[ \left( h(Y_2, \ldots, Y_d) e^{Y_1} - K \right)^+
    \right] \label{eq:call-Y} 
\end{align}
with
\begin{equation}
  \label{eq:h-def}
  h(\bary) = h(y_2, \ldots, y_d) \coloneqq \sum_{i=1}^d w_i \exp\left(
    \sum_{j=2}^d V_{i,j} y_j \right), \quad \bary \coloneqq (y_2, \ldots, y_d)
  \in \R^{d-1}.
\end{equation}

\begin{lemma}[{Conditional Expectation formula}]
  \label{lem:BS-formula}
  Let
  $\barY = (Y_2, \ldots, Y_d) = ((V^{-1}X)_2, \ldots, (V^{-1}X)_d)
  \sim \mathcal{N}\left(0, \barD\right)$,
  $\barD \coloneqq \diag(\lambda_2^2, \ldots,
  \lambda_d^2)$. Then
  \begin{equation*}
    E\left[\left. \left( \sum_{i=1}^d w_i e^{X_i} - K \right)^+ \right| \barY
    \right] = C_{BS}\left(h(\barY) e^{\lambda_1^2/2}, K, \lambda_1 \right),
  \end{equation*}
  where
  \begin{gather*}
    C_{BS}(S_0, K, \sigma) \coloneqq \Phi(d_1) S_0 - \Phi(d_2) K, \\
    d_{1/2} \coloneqq \f{1}{\sigma} \left[ \log\left( \f{S_0}{K} \right) \pm
      \f{\sigma^2}{2} \right],
  \end{gather*}
  is the Black-Scholes formula for $r=0$, with maturity $T = 1$.
\end{lemma}
\begin{proof}
  As $Y_1$ and $\barY$ are independent and $Y_1 \sim \mathcal{N}(0,
  \lambda_1^2)$, we have
  \begin{equation*}
    E\left[ \left. \left( h(Y_2, \ldots, Y_d) e^{Y_1} - K \right)^+
    \right| \barY = \bary \right] = E\left[ \left( h(\bary) e^{\lambda_1 Z} -
      K \right)^+ \right]
  \end{equation*}
  for some $Z \sim \mathcal{N}(0,1)$. On the other hand, for $r=0$ and
  maturity $T = 1$, the Black-Scholes formula is given by
  \begin{equation*}
    C_{BS}(S_0, K, \sigma) = E\left[ \left( S_0 e^{-\half \sigma^2 + \sigma Z}
      - K \right)^+ \right] = \Phi(d_1) S_0 - \Phi(d_2) K,
  \end{equation*}
  since $S_T = S_0 \exp\left( - \half \sigma^2 T + \sigma B_T \right)$ for a
  Brownian motion $B$. By comparing these expressions, we see that we have to
  choose $K = K$, $\sigma = \lambda_1$ and
  $S_0 = h(\bary) e^{\half \lambda_1^2}$.
\end{proof}

Lemma~\ref{lem:BS-formula} directly implies
\begin{proposition}
  \label{multivariate-BS-smoothing}
  The basket option price in the multivariate Black-Scholes setting satisfies
  \begin{equation}
   \label{eq:basket-int-problem}
    C_{\mathcal{B}} = E\left[ C_{BS}\left(h\left( \sqrt{\barD} Z \right)
      e^{\lambda_1^2/2}, K, \lambda_1 \right) \right], \quad Z \sim
  \mathcal{N}\left(0, I_{d-1} \right), \  \sqrt{\barD} =
  \diag(\lambda_2, \ldots, \lambda_d).
  \end{equation}
\end{proposition}

On the left-hand side of Figure \ref{fig:kinksmooth1}, a visualization of the integrand 
in \eqref{eq:call-gauss} before smoothing can be found while the corresponding smoothed integrand 
from \eqref{eq:basket-int-problem} is presented on the right-hand side of Figure \ref{fig:kinksmooth1}. 

\begin{figure}[htb]
\includegraphics[width=.48\textwidth]{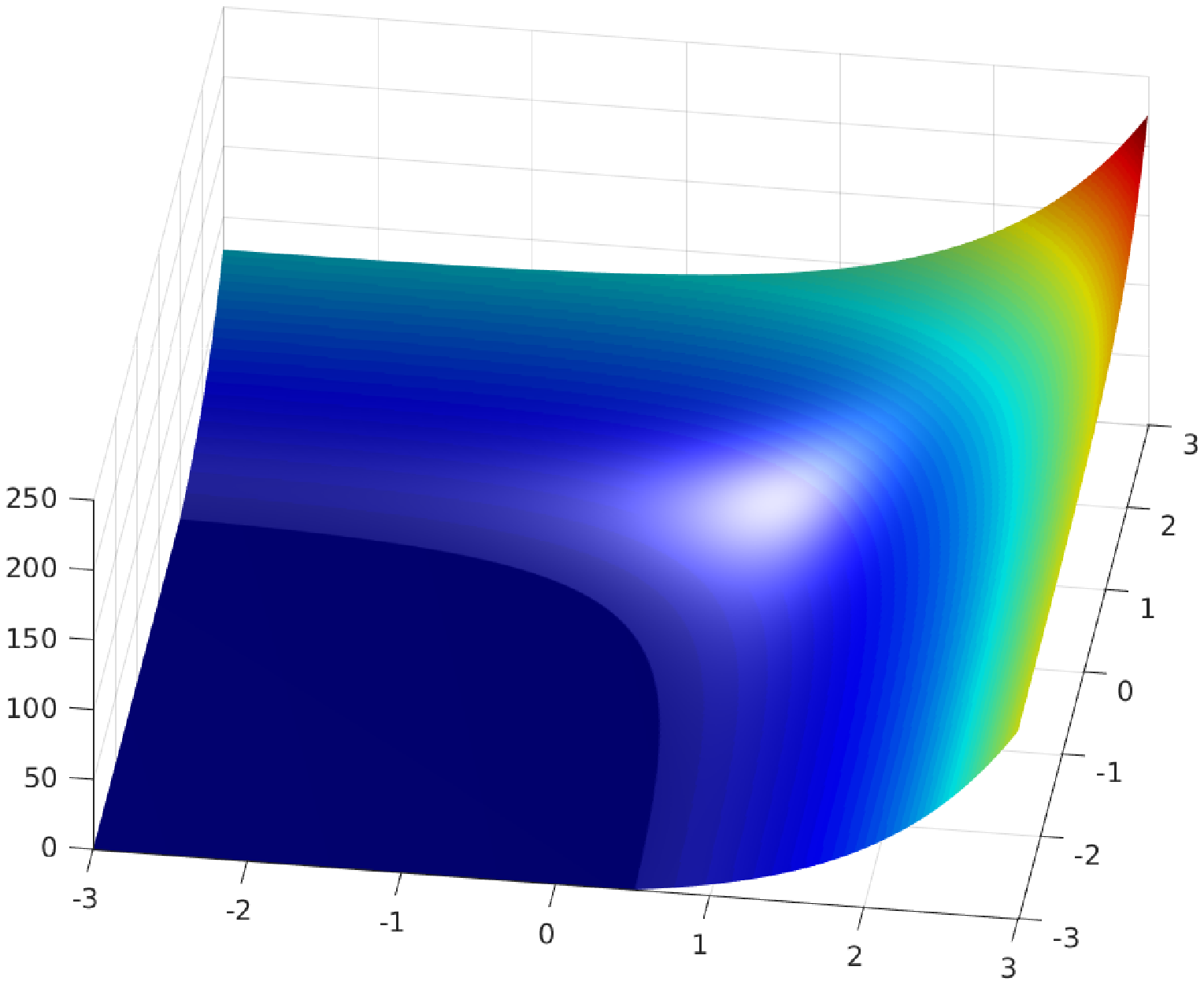}
\includegraphics[width=.48\textwidth]{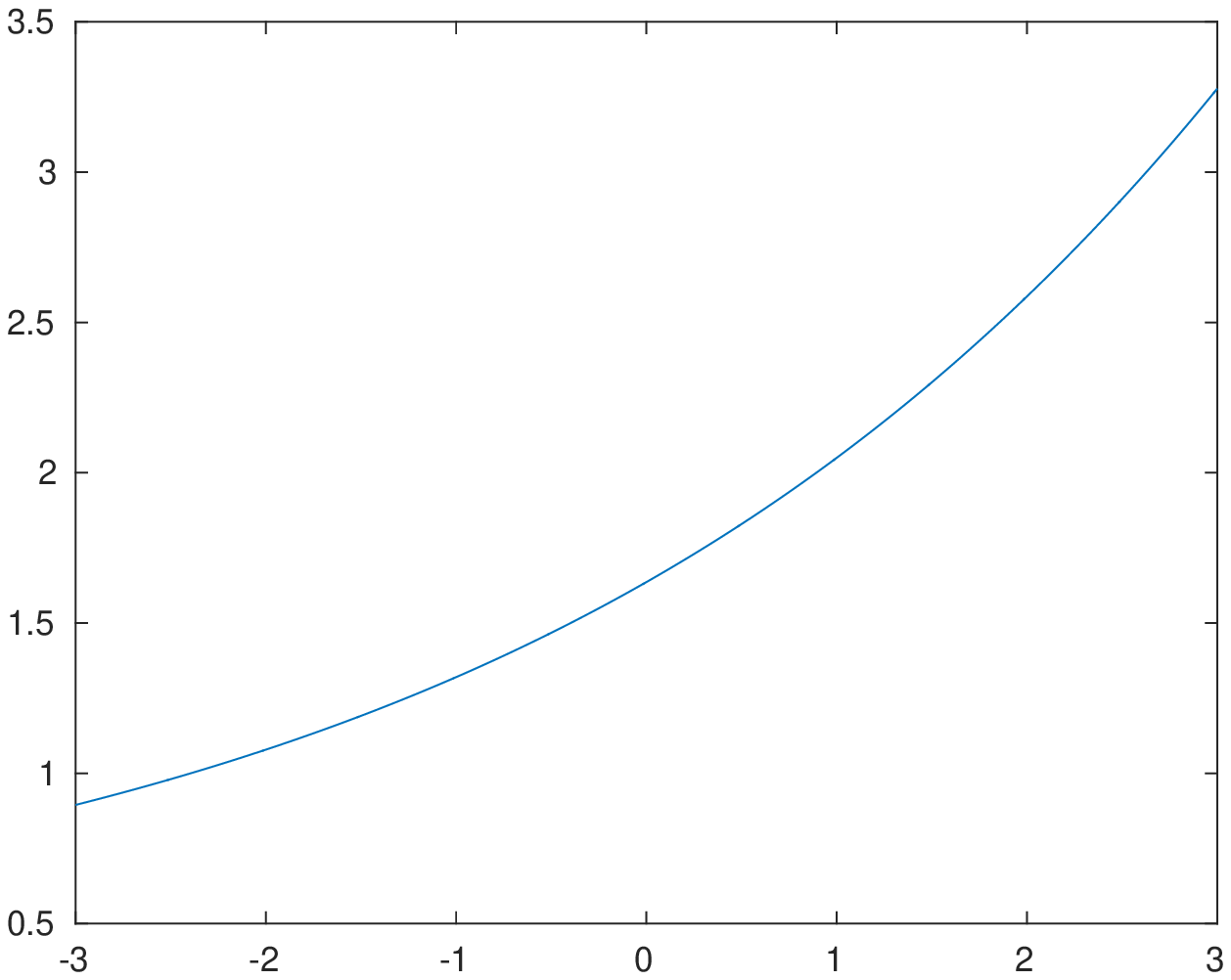}
\caption{\label{fig:kinksmooth1} 
An example of a two-dimensional integrand with kink from \eqref{eq:call-gauss} and 
its smoothed counterpart with respect to \eqref{eq:basket-int-problem}.}
\end{figure}

%\begin{remark}\label{rem:BS-formula-01}
  As remarked earlier, a similar closed-form expression cannot be obtained
  when the first column $V_{\cdot,1}$ of the matrix $V$ is a general
  $d$-dimensional vector. However, we may still get an explicit formula if
  $V_{\cdot,1}$ only takes values in $\set{0,1}$. For simplicity, let us
  assume that the first $k$ entries of $V_{\cdot,1}$ are $1$ and the
  remaining entries are $0$. The computation before Lemma~\ref{lem:BS-formula}
  then gives
  \begin{gather*}
    C_{\mathcal{B}} = E\left[ \left( h_1(Y_2, \ldots, Y_k) e^{Y_1} + h_2(Y_{k+1},
        \ldots, Y_d) - K \right)^+ \right],\\
    h_1(y_2, \ldots, y_k) \coloneqq \sum_{i=1}^k w_i \exp\left( \sum_{j=2}^d
      V_{i,j} y_j \right), \\
    h_2(y_{k+1}, \ldots, y_d) \coloneqq \sum_{i=k+1}^d w_i \exp\left( \sum_{j=2}^d
      V_{i,j} y_j \right).
  \end{gather*}
  By conditioning again on $\barY$, we once again arrive at the Black-Scholes
  formula, this time requiring a shift in $K$ as well. In the end, we obtain
  \begin{equation}
  \label{eq:CSwithv}
    E\left[\left. \left( \sum_{i=1}^d w_i e^{X_i} - K \right)^+ \right| \barY
    \right] = C_{BS}\left(h_1(Y_2, \ldots, Y_k) e^{\lambda_1^2/2}, K -
      h_2(Y_{k+1}, \ldots, Y_d), \lambda_1 \right),
  \end{equation}
  in the sense that
  \begin{equation*}
    C_{BS}(S_0, K, \sigma) = S_0 - K \text{ for } K < 0. 
  \end{equation*}
  In general, we therefore suggest to choose $V_{\cdot,1}$ such as to maximize
  the effective smoothing parameter $\lambda_1$.
%\end{remark}

%\blue{
More concretely, let us assume that the eigenvalues $\mu_1^2 \ge \cdots \ge
\mu_d^2 > 0$ of $\Sigma$ are given. Let $D = \diag\left(\mu_1^2, \ldots,
  \mu_d^2\right)$. Of course, the matrix $Q \in O(d)$ of corresponding
eigenvectors of $\Sigma$ satisfies 
\begin{equation*}
  \Sigma = Q D Q^\top.
\end{equation*}
Denoting $\lambda_1^2 = \lambda_1^2(D,Q,v) = \ip{v}{\Sigma^{-1}v}^{-1}$ and
$\mathcal{V} \coloneqq \set{0,1}^d \setminus \set{0}$, we are looking for the
worst possible smoothing effect given the eigenvalues $D$ of the covariance
matrix $\Sigma$ and \emph{given that we choose the vector $v$ optimally},
i.e., we are looking for
\begin{equation*}
  \lambda_\ast^2(D) \coloneqq \min_{Q \in O(d)} \max_{v \in \mathcal{V}}
  \lambda_1^2(D,Q,v).
\end{equation*}
\begin{lemma}\label{lem:lower_bound}
  We have $\lambda_\ast^2(D) \ge \mu_d^2$.
\end{lemma}
\begin{proof}
  We obviously have
  \begin{equation*}
    \lambda_\ast^2(D) = \min_{Q \in O(d)} \max_{v \in \mathcal{V}}
    \lambda_1^2(D,Q,v) \ge \max_{v \in \mathcal{V}} \min_{Q \in O(d)}
    \lambda_1^2(D,Q,v) \ge \min_{Q \in O(d)} \lambda_1^2(D,Q,e_d),
  \end{equation*}
  for $e_d = (0, \ldots, 0, 1)$. 

  Clearly, the minimizing $Q$ should have $e_d$ as its last row, making $e_d$
  the eigenvector corresponding to the smallest eigenvalue of $\Sigma$. 
  Indeed, for any $Q \in O(d)$ we have
  \begin{equation*}
    \lambda_1^2(D,Q,e_d) = \ip{Q^\top e_d}{D^{-1}Q^\top e_d}^{-1} =
    \ip{q}{D^{-1}q}^{-1} = \left(\sum_{i=1}^d \f{q_i^2}{\mu_i^2} \right)^{-1}
    \eqqcolon f(D,q),
  \end{equation*}
  where $q = q(Q)$ denotes the last row of $Q$ (understood as column
  vector). Note that the range of $q$ (as a function of $Q$) is $S_{d-1}$,
  hence we need to minimize $f(D,q)$ of all vectors $q$ with $q_1^2 + \cdots +
  q_d^2 = 1$. It is easy to see that the minimizer is $q = e_d$ and the value
  is
  \begin{equation*}
    \min_{Q\in O(d)} \lambda_1^2(D,Q,e_d) = \mu_d^2.
  \end{equation*}

  In fact, we claim that the above lower bound holds uniformly over $v$ in the
  sense that
  \begin{equation*}
    \max_{v \in \mathcal{V}} \min_{Q \in O(d)} \lambda_1^2(D,Q,v) = \mu_d^2.
  \end{equation*}
  The reason is that for arbitrary $v \in \mathcal{V}$ the minimizing $Q$ is
  still given such that $v$ is (up to normalization) the eigenvector
  corresponding to the smallest eigenvalue of $\Sigma$. Hence,
  \begin{equation*}
    \min_{Q \in O(d)} \lambda_1^2(D,Q,v) = \left( \f{\abs{v}^2}{\mu_d^2}
    \right)^{-1} = \f{\mu_d^2}{\abs{v}^2}.
  \end{equation*}
  The equality follows by noting that $\min_{v\in\mathcal{V}} \abs{v}^2 = 1$.
\end{proof}

\begin{remark}
  It is easy to check that the inequality in Lemma~\ref{lem:lower_bound} is
  generally strict. We are not aware of more explicit expressions for
  $\lambda_\ast^2(D)$.
\end{remark}
%}

% \begin{remark}
% \label{rem:decder}
%   Generally speaking, the decay of derivatives of the integrand
%   in~\eqref{eq:basket-int-problem} depends -- inter alia -- on the size of
%   $\lambda_1^2 = \ip{\mathbf{1}}{\Sigma^{-1} \mathbf{1}}^{-1}$. If 
%   we normalize the variances of the individual components,
%   then $\lambda_1^2$ mostly depends on the angle of $\mathbf{1}$ with the
%   eigenspace corresponding to the largest eigenvalue of $\Sigma$.
% \end{remark}

\begin{remark}
\label{rem:impsampling}
It is worth observing that after the conditional expectation
\eqref{eq:basket-int-problem} one may also perform a change of measure on the
resulting $d-1$ dimensions to enhance the convergence of all the quadratures
discussed in this work. For instance, this is particularly important for OTM
options.
\end{remark}

\section{Numerical example 1: Multivariate Black-Scholes setting}
\label{sec:mult-black-schol}

In our first numerical example, we consider the pricing problem
\eqref{eq:basket-option} of a European basket option in a Black-Scholes
model. This price depends on the strike price \(K\), the weight vector \(c\)
and the vector \(S_{T}\) containing the values of the different assets at the
maturity \(T\). Moreover, the distribution of \(S_{T}\) can be deduced from
the initial values of the assets \(S_0\), the vector of volatilities
\(\sigma\) and the correlation matrix \(\rho\), which determine the
Black-Scholes model in \eqref{eq:BS-SDE}. The initial values in our examples
are chosen randomly, that is independently and uniformly distributed from the
interval $S_0^i \in [8,20]$. The volatilities are also chosen randomly
from the interval $\sigma_i \in [0.3,0.4]$. Following \cite{Doust}, the
correlation matrix \(\rho= \tau \tau^{\top}\) is given by a lower-triangular
matrix \(\tau\), {parameterized by a vector $x \in [-1,1]^{d-1}$}
as follows:
{
\[
 \tau_1 = \begin{pmatrix} 1  \\
  \operatorname{cp}(x) \end{pmatrix}, \quad 
  \tau_2 = \sqrt{1-x_1^2}\begin{pmatrix} 0  \\1 \\
  \operatorname{cp}(x_{2:d-1}) \end{pmatrix},\ldots ,\quad
  \tau_{d} = \sqrt{1-x_{d-1}^2}\begin{pmatrix} 0  \\\vdots \\0\\
  1\end{pmatrix}. 
\]}
Herein, we employed the MATLAB-inspired notation
\(x_{2:d-1}=[x_2,\ldots,x_{d-1}]^{\top}\). In addition, we denote by
\(\operatorname{cp}\colon \R^{d-1}\to \R^{d-1}\) the cumulative product given
by
\[
\operatorname{cp}(x)=[x_1,x_1x_2,\ldots, x_1x_2\cdots x_{d-1}]^{\top}. 
\]
{Note that any such matrix is a proper correlation matrix,
  which only depends on $d-1$ (instead of $d(d-1)/2$) free parameters, and is
  therefore much easier to apply in practice. For sake of concreteness, we
  choose independently and uniformly distributed entries \(x_i\in [0.8,1]\),
  which lead to positive correlations between the individual assets comprising
the basket, a typical situation in equity markets.}
The weight vector \(c\) is chosen such that the basket is an average of the
different assets, i.e.~\(c_i=1/d\). Moreover, we choose three different
settings for the strike price, \(K = c^{\top}S_0\) (ATM), \(K = 1.2\cdot
c^{\top}S_0\) (OTM) and \(K = 0.8\cdot c^{\top}S_0\) (ITM).

\begin{remark}
  We tested our experiments with different, randomly chosen weight vectors
  \(c\in [0,1]^d\) and obtained similar results. Hence, it seems that there is
  only a slight dependence between the weight vector in the basket and the
  performance of the different quadrature methods.
\end{remark}

We compare several integration schemes applied to the original problem
\eqref{eq:call-gauss} and the smoothened problem
\eqref{eq:basket-int-problem}. To be more precise, we consider the MC method,
the QMC method based on Sobol points and the sparse-grid method described in
Section 2. {Moreover, we also use a sparse grid approximation to the smoothed integrand 
as control variate to accelerate the convergence of MC and QMC. 
To keep track of the different quadratures that are used in the numerical examples, we 
summarize them together with their acronyms in table \ref{tab:acro}. Additionally, the 
colors and markers with which they appear in the convergence plots 
are listed here. }

\begin{table}[hbt]
\begin{tabular}{c|c|c|c}
 Method & Acronym & Color & Marker\\
\hline 
Adaptive sparse grid to \eqref{eq:call-gauss}& aSG 
&\begin{tikzpicture}
\fill[darkblue] (-.4,-.1) rectangle (.4,.1); %darkblue
\end{tikzpicture}
 & $\circ$  \\
\hline
Quasi-Monte Carlo to \eqref{eq:call-gauss}& QMC  
&\begin{tikzpicture}
\fill[darkgreen] (-.4,-.1) rectangle (.4,.1); % darkgreen
\end{tikzpicture}
 & $\circ$  \\
\hline
Monte Carlo to \eqref{eq:call-gauss}& MC 
&\begin{tikzpicture}
\fill[darkred] (-.4,-.1) rectangle (.4,.1); % darkred
\end{tikzpicture}
 &
 $\circ$  \\
\hline
aSG to \eqref{eq:basket-int-problem}& aSG+CS 
&\begin{tikzpicture}
\fill[cyan] (-.4,-.1) rectangle (.4,.1); % cyan
\end{tikzpicture}
&
$\star$ \\
\hline
aSG with respect to \eqref{eq:CSwithv}& aSG+CS2 
&\begin{tikzpicture}
\fill[royblue] (-.4,-.1) rectangle (.4,.1); % cyan
\end{tikzpicture}
&
$\star$ \\
\hline
QMC to \eqref{eq:basket-int-problem}& QMC+CS 
&\begin{tikzpicture}
\fill[green] (-.4,-.1) rectangle (.4,.1); % cyan
\end{tikzpicture}
&
$\star$ \\
\hline
MC to \eqref{eq:basket-int-problem}& MC+CS 
&\begin{tikzpicture}
\fill[red] (-.4,-.1) rectangle (.4,.1); % cyan
\end{tikzpicture}
&
$\star$ \\
\hline
QMC to \eqref{eq:basket-int-problem} with control variate & QMC+CS+CV 
&\begin{tikzpicture}
\fill[sprgreen] (-.4,-.1) rectangle (.4,.1); % cyan
\end{tikzpicture}
&
$\diamond$ \\
\hline
MC to \eqref{eq:basket-int-problem} with control variate & MC+CS+CV 
&\begin{tikzpicture}
\fill[viol] (-.4,-.1) rectangle (.4,.1); % cyan
\end{tikzpicture}
&
$\diamond$ \\
\end{tabular}
\caption{\label{tab:acro} Different quadrature methods, their acronyms, colors 
and markers.}
\end{table}

\subsection{Performance of the sparse-grid methods}

\begin{figure}[htb]
\includegraphics[width=.48\textwidth]{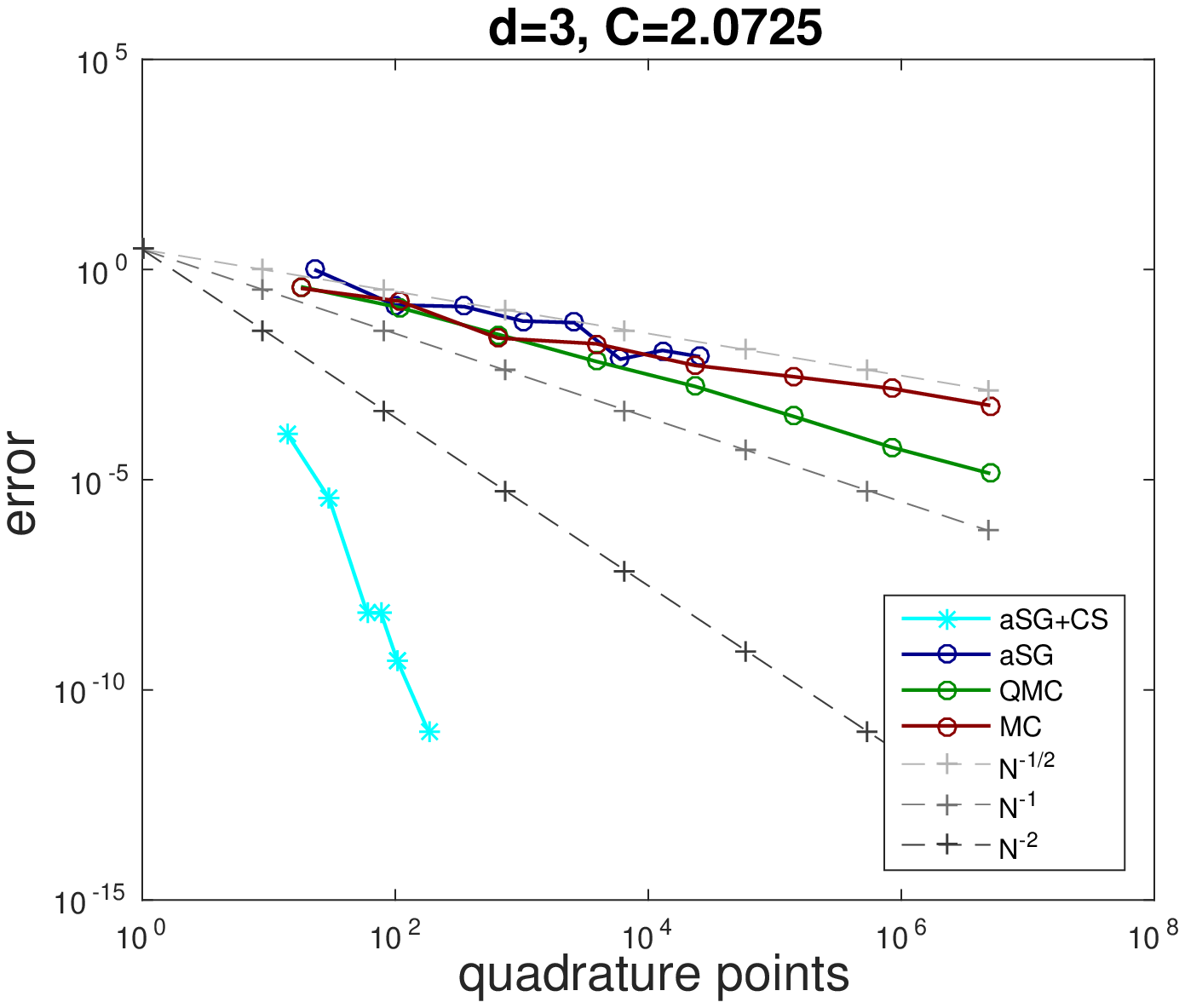}
\includegraphics[width=.48\textwidth]{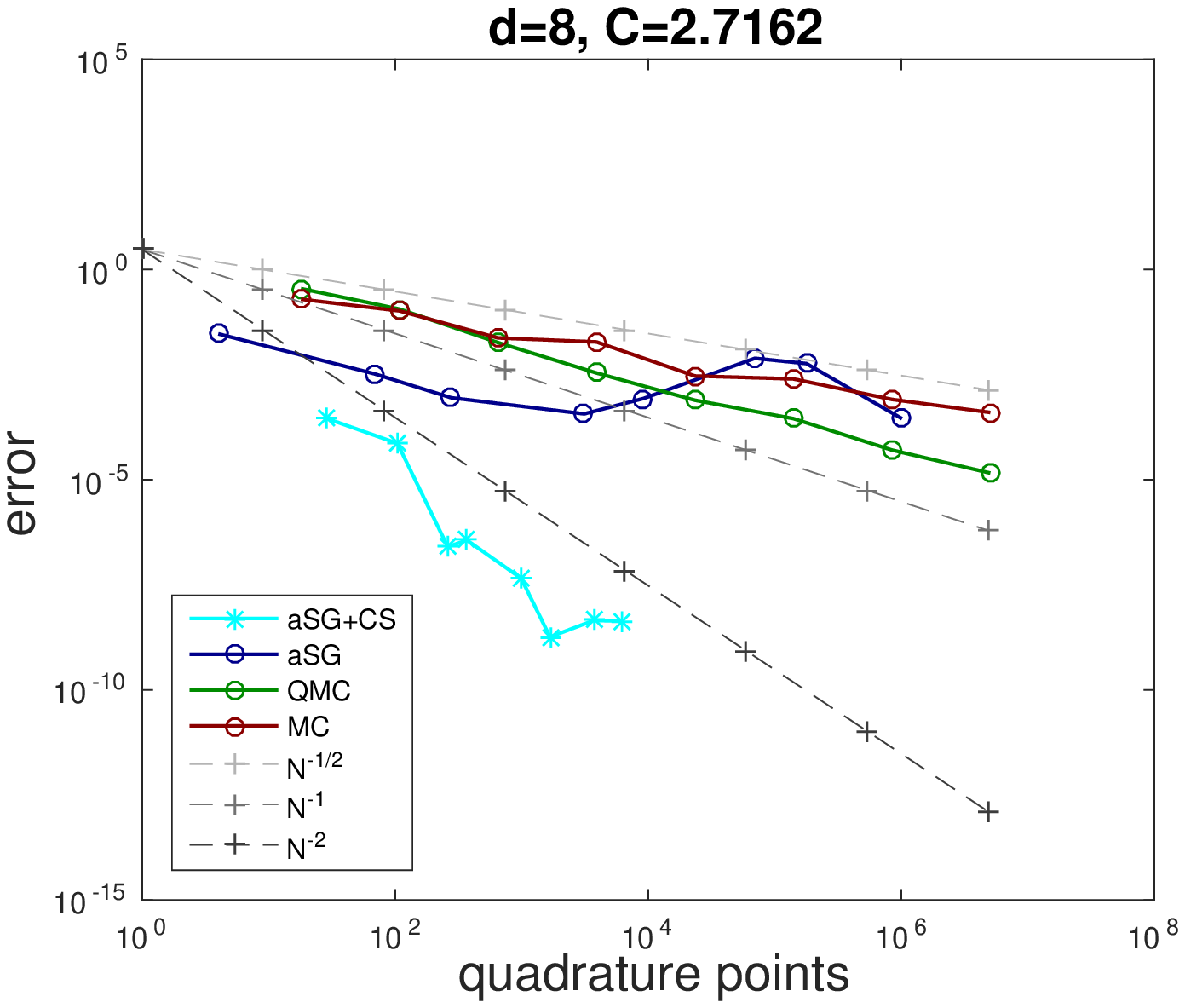}
\includegraphics[width=.48\textwidth]{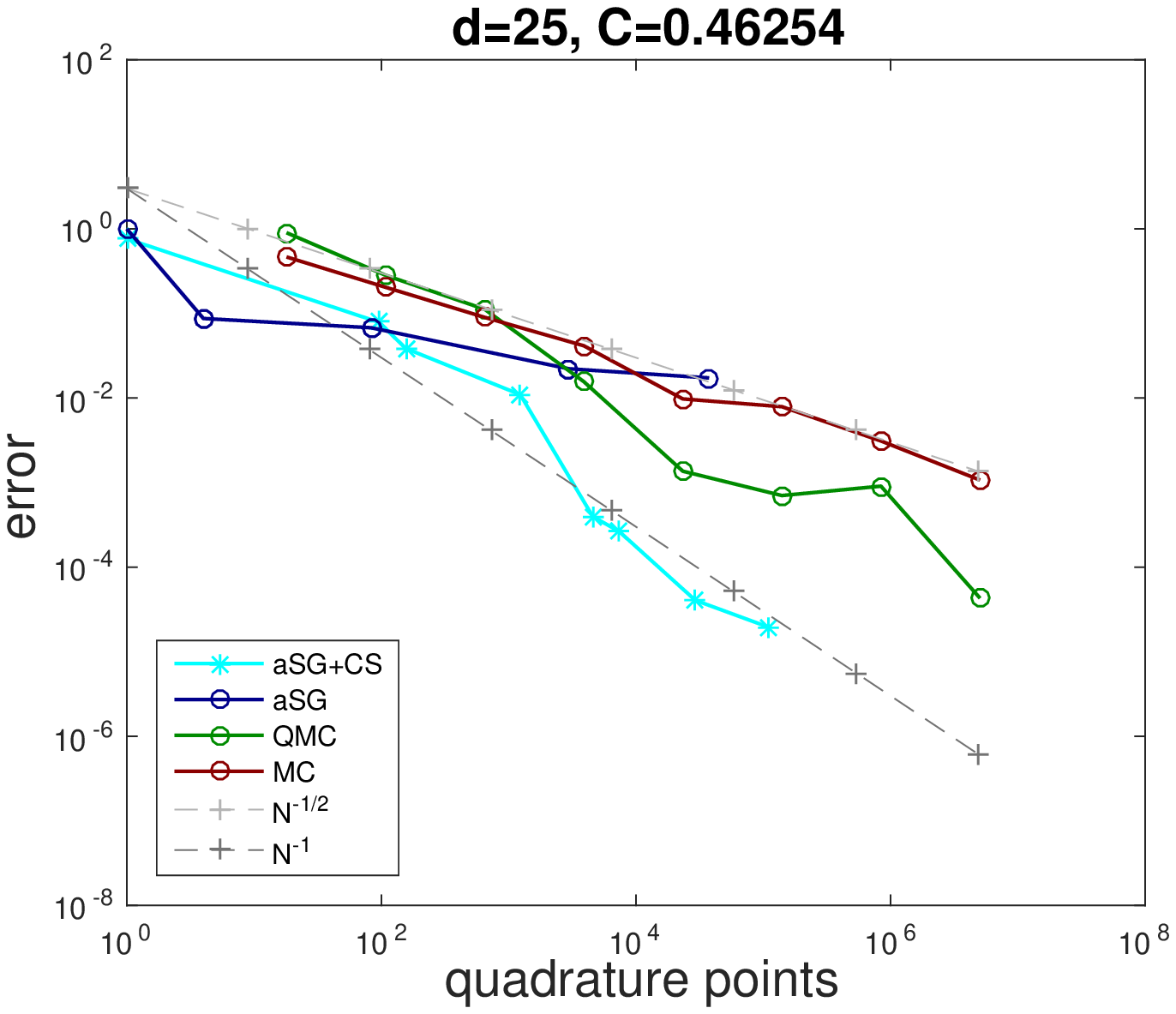}
\caption{\label{fig:vol3+1}Errors for \(d=3\), \(d=8\) and \(d=25\) with volatilities 
selected randomly from the interval \([0.3,0.4]\).}
\end{figure}

In this subsection, we investigate the convergence behaviour of the adaptive
sparse-grid method for the smoothened problem \eqref{eq:basket-int-problem}
(aSG+CS). Therefore, we apply the (aSG+CS) to our model problem in dimension
\(d=3\) in the ATM case, \(d=8\) in the ITM case and \(d=25\) in the OTM
case. {Note that the ITM case is the easiest case and the OTM 
case is the hardest for numerical computation. The results would slightly 
improve when considering the ITM case in $d=25$ dimensions. However, we would like to 
demonstrate that the smoothing works well even for the hardest case in moderately 
high dimensions.} 
As a reference solution, we use an adaptive sparse-grid quadrature to
determine \eqref{eq:basket-int-problem} with a very small tolerance,
i.e.~\(\varepsilon =10^{-11}\) for \(d=3\), \(\varepsilon =10^{-9}\) for
\(d=8\), and \(\varepsilon = 10^{-7}\) for \(d=25\) respectively. We use the
listed Genz-Keister points from \cite{HW08} as the sequence of underlying
univariate quadrature points. Unfortunately, there exist only nine different
Genz-Keister extensions and it might happen that a higher precision is needed
in a particular direction. In this case, we use Gau{\ss}-Hermite quadratures with a
successively higher degree of precision {for the consecutive 
members of the quadrature sequence}. The 1D Gau{\ss}-Hermite
points and weights can easily be constructed for an arbitrary degree of
precision by solving an associated eigenvalue problem; see \cite{GW69}
for the details.

To observe the convergence behaviour of the aSG+CS, we successively refine the
tolerance (e.g., from \(10^{-2}\) to \(10^{-9}\) for \(d=3\)) and compute the
relative error between the corresponding approximation to
\eqref{eq:basket-int-problem} and the reference solution.  To compare the
results with other methods and also to validate the reference solution, we
also apply an MC quadrature, a QMC quadrature and an adaptive sparse-grid
quadrature (SG) to the original problem \eqref{eq:call-gauss} and compare the
results with the reference solution as well.  Herein, we increase the number
of quadrature points for the (quasi-) Monte Carlo quadrature as \(3\cdot 6^q\)
for \(q=1,\ldots,8\) {which is adjusted for the sake of comparison with the 
aSG+CS method for the $25$-dimensional example}. In addition, we use \(20\) runs 
of the Monte-Carlo estimator on each level \(q\) and plot the median of the relative 
errors to the reference solution of these \(20\) runs. 

The results for \(d=3\), \(d=5\) and \(d=25\) are depicted in
Figure~\ref{fig:vol3+1}. As expected, the MC quadrature converges in each
dimension algebraically with a rate \(1/2\) to the reference solution,
while the rate of the QMC quadrature is close to \(1\), {despite the 
non-smooth integrand}. 
%\textcolor{brown}{
%which is surprising since the dimensionality seems not to affect the convergence 
%rate and since the integrand does not have bounded variation in the 
%sense of Hardy and Krause. Indeed, the integrand tends to infinity when any 
%coordinate tends to infinity, however, the Gaussian density tends much faster 
%to zero. ??? Not sure if we should mention this. Completely agree that we should 
%not mention this.}
 The convergence of
the aSG is comparable to that of the MC for \(d=3\) and becomes worse for
\(d=8\) and \(d=25\). Hence, it is not very suitable to tackle the original
problem \eqref{eq:call-gauss} with aSG. In contrast to that, the aSG+CS outperforms all
other considered methods, especially for \(d=3\) and \(d=8\), in both
convergence rate and constant. For \(d=3\), the rate is exponential rather
than algebraic and the observed algebraic rate for \(d=8\) is \(2\). In
\(d=25\) dimensions, the rate deteriorates to \(1\) but the constant is still
around a factor \(35\) less than that of QMC.
\begin{remark}
  The convergence results are shown in terms of quadrature points. 
  {In case of adaptive sparse grids it might happen that the same 
  quadrature point appears multiple times since we evaluate tensor products of difference 
  quadrature formulas associated with multi-indices $\boldsymbol{\alpha}$. For example, 
  for the approximation with tolerance $10^{-9}$ for d=$3$, the aSG+CS method requires 
  $183$ quadrature points, but only $64$ of them are distinct, cf.~Figure \ref{fig:grid}. 
  Hence, the convergence results for adaptive sparse grid methods in terms of function 
  evaluations at quadrature points could be improved if the function evaluations would be 
  stored.} 
  However, it is also
  interesting to compare the computational times to see the overhead of the
  adaptive sparse-grid construction. In Table~\ref{tab:times}, we depict
  computational times in seconds and errors for the different quadrature
  methods at a comparable number of quadrature points for each dimension.  As
  we can deduce from these times, there is indeed a huge overhead for the
  aSG+CS. In dimension \(d=25\), for example, the computation of the adaptive
  sparse-grid method with around \(25 \%\) more quadrature points requires
  around \(23\) times the computation time in comparison to QMC. Nevertheless,
  the error of the aSG+CS is around a factor \(600\) smaller than that of QMC.
\begin{table}[hbt]
\begin{adjustbox}{max width=\textwidth}
\begin{tabular}{c|c|c|c|c|c|c|c|c|c}
& \multicolumn{3}{|c|}{aSG+CS} & \multicolumn{3}{|c|}{QMC} & \multicolumn{3}{|c}{MC}\\
\hline 
& time & error & points  & time & error & points & time & error & points\\
\hline
\(d=3\) & 0.0057 & 4.9 e-10 & 104 & 0.0016 & 1.25 e-1 & 108 & 0.0013 & 1.77 e-1 & 108\\
\hline
\(d=8\) & 0.3675 & 1.81 e-9 & 24622 & 0.0161 & 5.39 e-3 & 23328 & 0.0135 & 1.38 e-2 & 23328\\
\hline
\(d=25\) & 5.4283 & 1.04 e-6 & 174098 & 0.2409 & 6.18 e-4 & 139968 & 0.2188 & 1.29 e-3 & 139968\\
\end{tabular}
\end{adjustbox}
\caption{\label{tab:times}Computation times for the different quadrature methods}
\end{table}
Note that all the computations are done in MATLAB and that the evaluation of
the integrand is completely vectorized in case of the (Q)MC. Naturally, this
is not possible for the adaptive sparse-grid quadrature, since we adaptively
add indices to the index set corresponding to difference quadrature rules with
a relatively low number of quadrature points.  Although the evaluation of the
integrand in each difference quadrature rule is vectorized, we need to do this
several times during the algorithm. Hence, a MATLAB implementation is not the
most efficient one for adaptive sparse-grid quadratures or adaptive methods in
general and the overhead could be reduced drastically with an efficient
implementation in C, for example.

{Finally, note that once an adaptive sparse grid has been
  constructed, it could potentially be re-used for pricing of options with
  similar parameters. In  that case, of course, the overhead of constructing
  the sparse grid disappears completely.}
\end{remark}

\begin{figure}[htb]

\includegraphics[width=.48\textwidth]{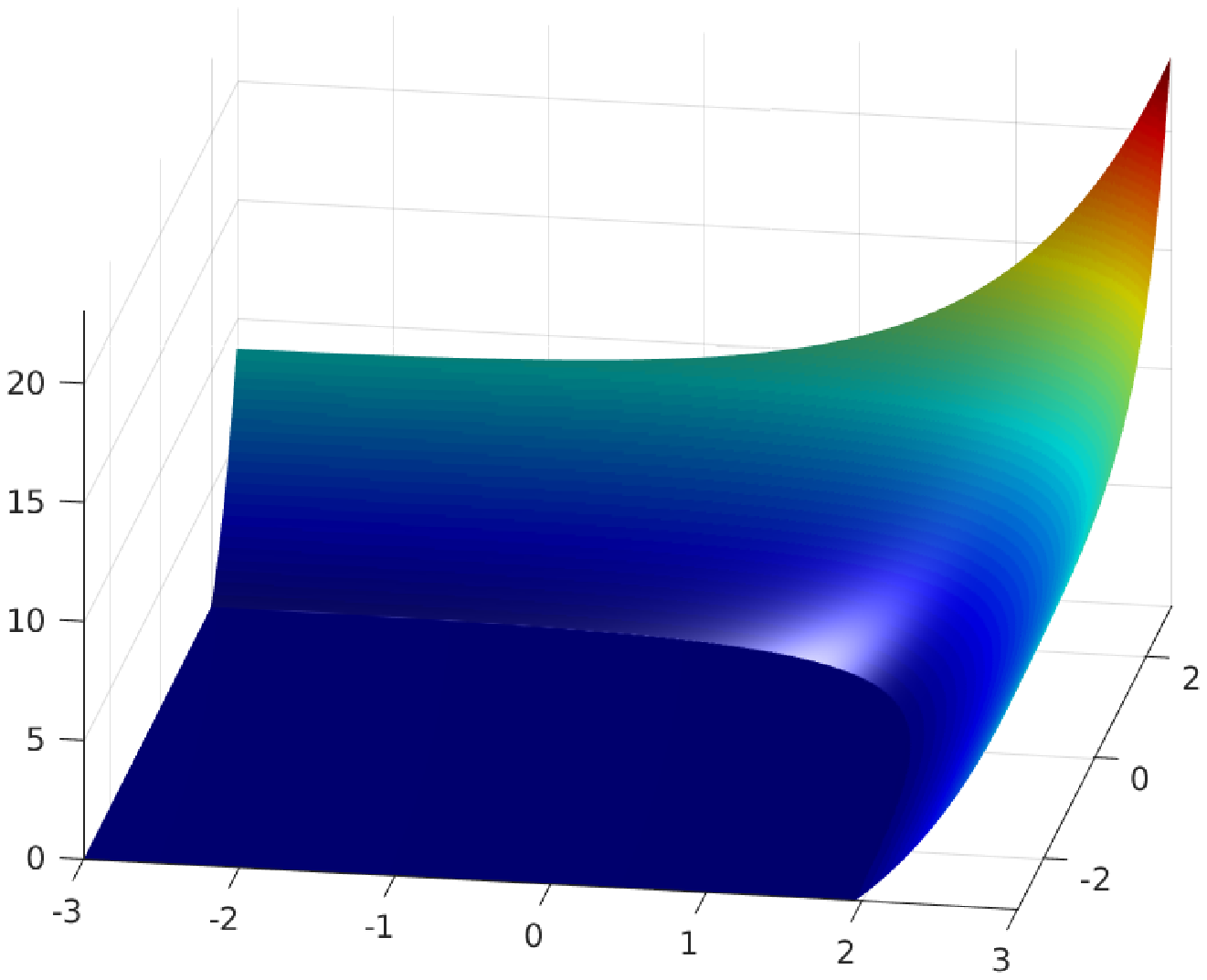}
\includegraphics[width=.48\textwidth]{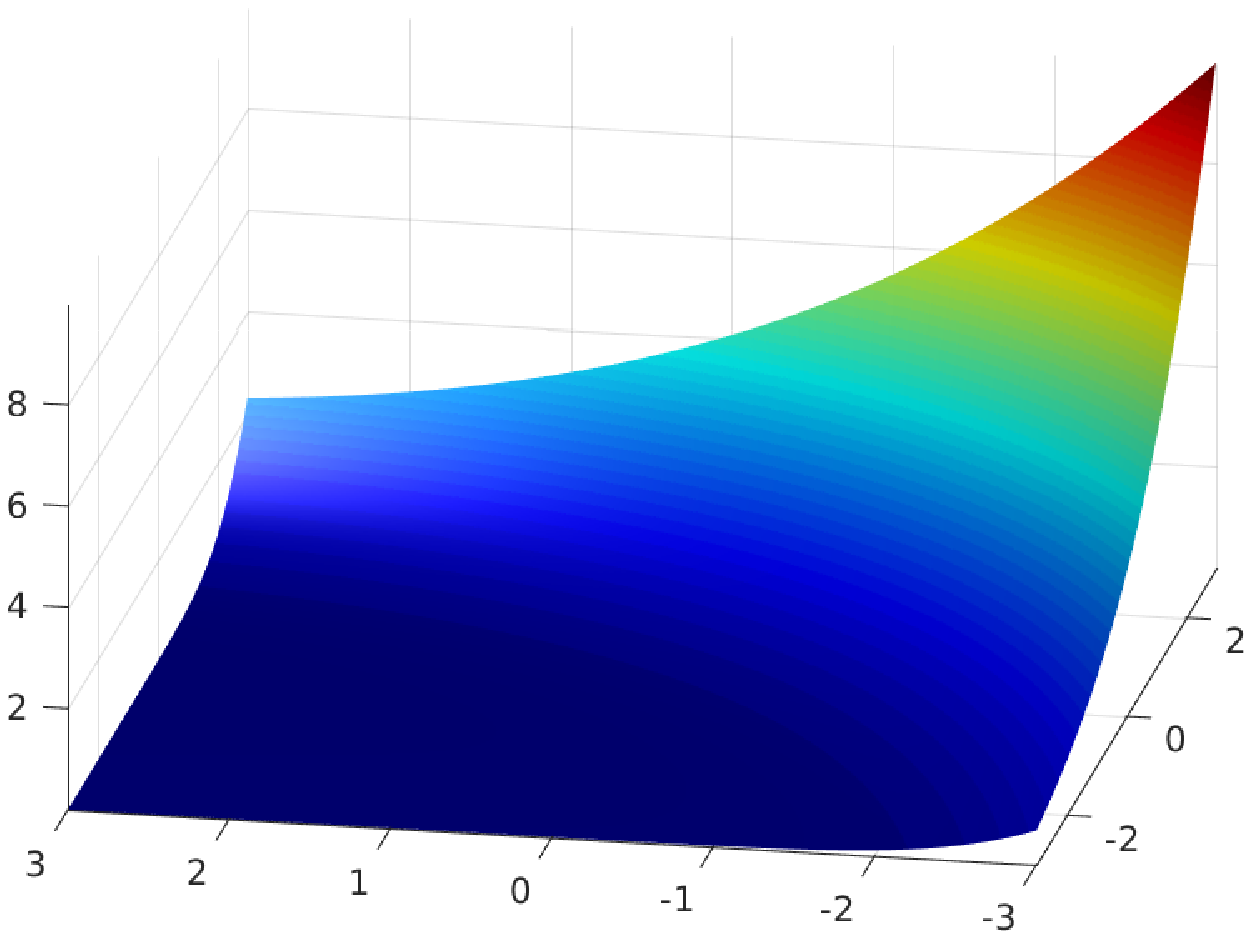}
\caption{\label{fig:kinksmooth2} 
Projection on the first two variables of the Integrands  in \eqref{eq:call-gauss} (left) and 
 \eqref{eq:basket-int-problem} (right) for \(d=25\) ( ``out of the money'').}
\end{figure}

In summary, we find that the adaptive sparse-grid quadrature applied to
\eqref{eq:basket-int-problem} accurately approximates the
value of a basket option. In particular, it significantly improves the
performance of the adaptive sparse-grid quadrature applied to
\eqref{eq:call-gauss}. This is due to the fact that the integrand in
\eqref{eq:basket-int-problem} is smooth while the integrand in
\eqref{eq:call-gauss} is not even differentiable. 
In order to corroborate the latter point, we illustrate in Figure \ref{fig:kinksmooth2} 
the projection on the first two variables of the \(8\)-dimensional 
integrand in \eqref{eq:call-gauss} and of the \(7\)-dimensional integrand 
in \eqref{eq:basket-int-problem} in the out of the money case. 
In addition to the smoothing effect, we further observe that the range of 
function values is reduced by around a factor $2.5$ for the latter integrand. 
Hence, it seems reasonable to additionally investigate the effects of 
the smoothing technique on the MC and QMC quadrature. 

\subsection{Smoothing effect for MC and QMC quadrature}

\begin{figure}[htb]
\includegraphics[width=.48\textwidth]{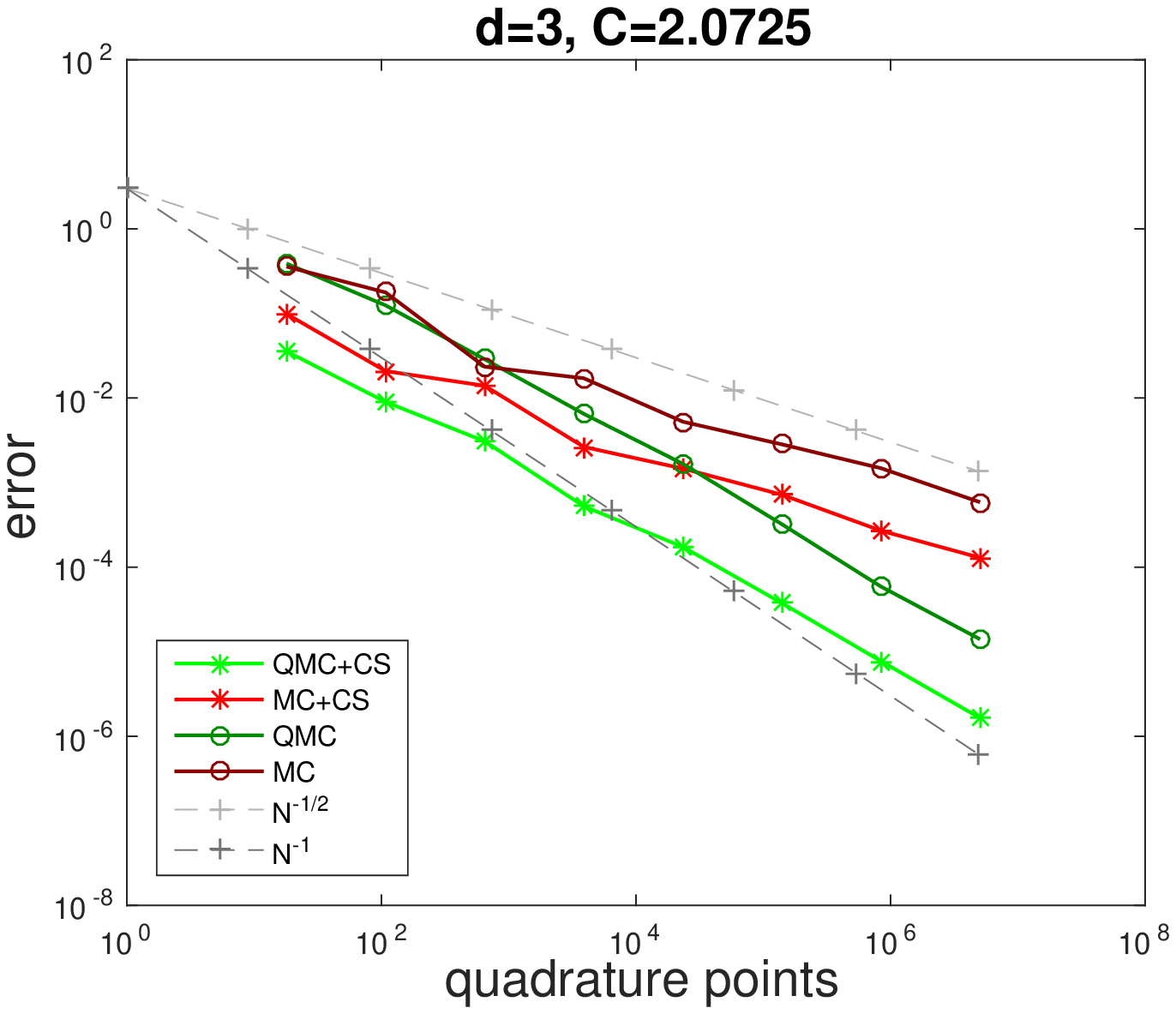}
\includegraphics[width=.48\textwidth]{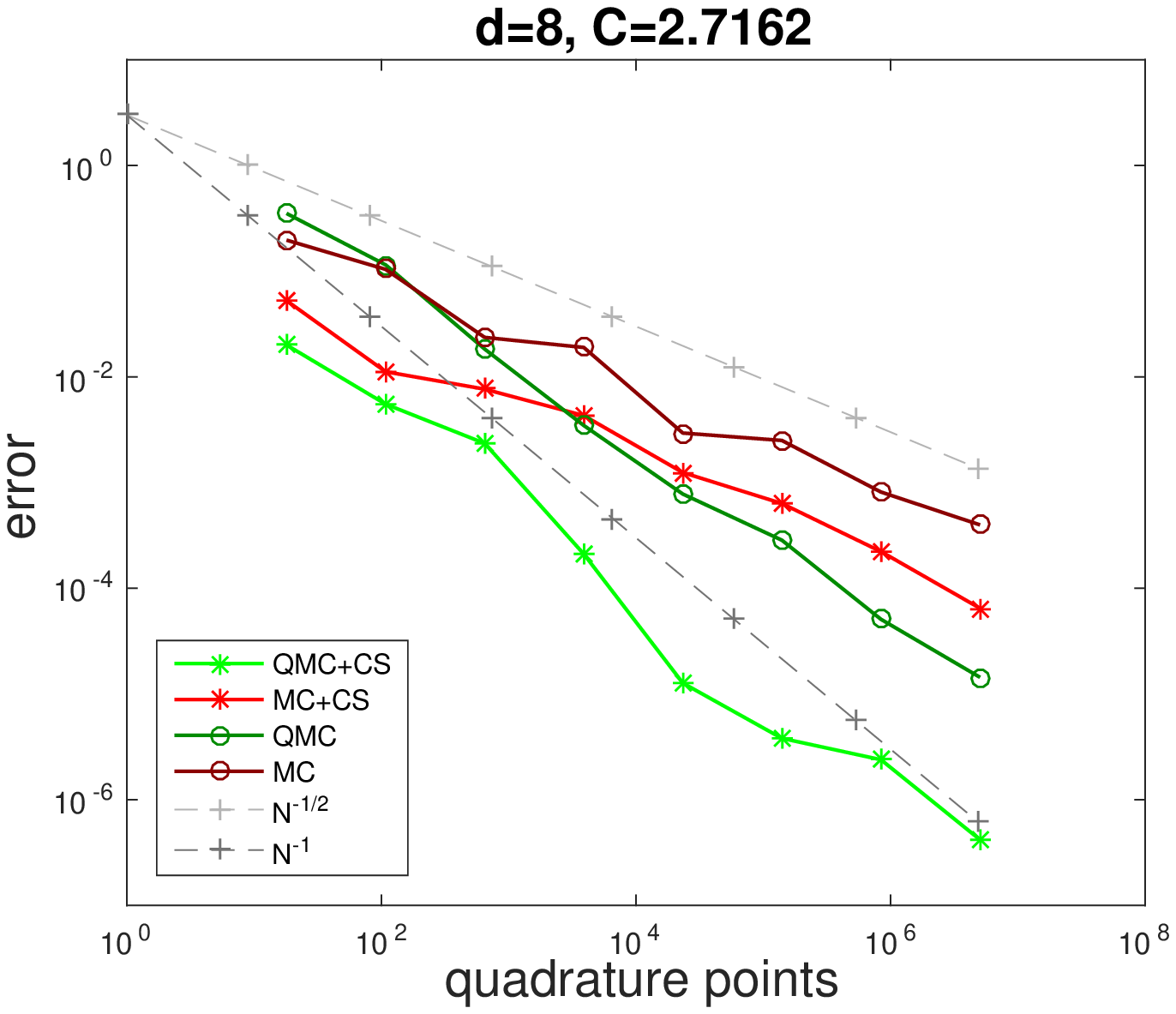}
\includegraphics[width=.48\textwidth]{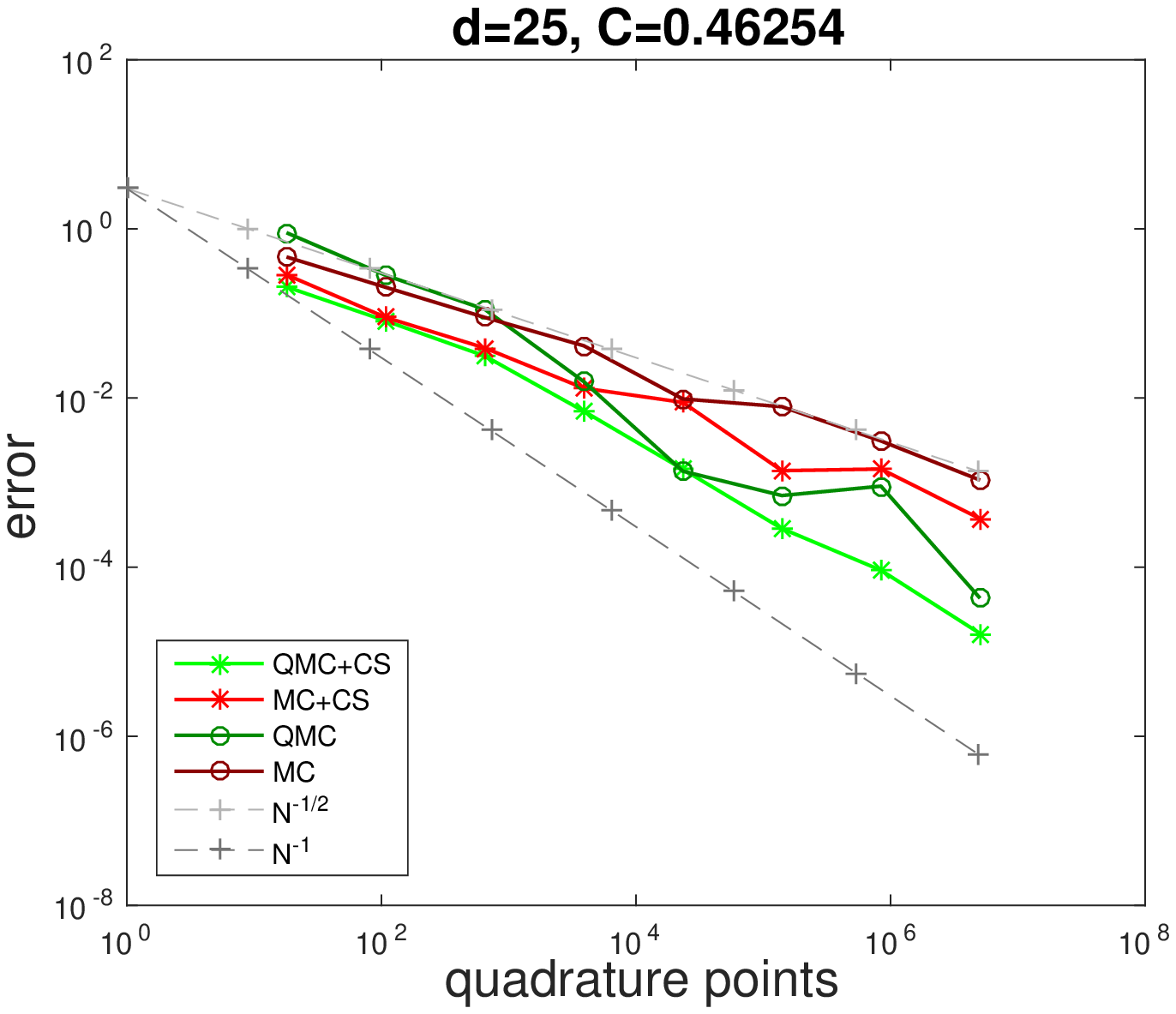}
\caption{\label{fig:vol1+3semc}Smoothing effect for the (Q)MC quadrature 
for \(d=3\), \(d=8\) and \(d=25\) with volatilities 
selected randomly from the interval \([0.3,0.4]\).}
\end{figure}

In this subsection, we examine the smoothing effect on the (Q)MC quadrature.
To that end, we apply the (Q)MC quadrature with the same number of quadrature
points as before (i.e.~\(3\cdot 6^q\) for \(q=1,\ldots,8\)) to approximate the
integral in \eqref{eq:basket-int-problem} and compare the results with those
of the (quasi-) Monte Carlo quadrature applied to \eqref{eq:call-gauss}. For
the Monte Carlo quadrature, we expect that the smoothing effect is not as
strong as for the sparse-grid quadrature. Nevertheless, the convergence
constant might be improved since we determined a conditional expectation to
deduce \eqref{eq:basket-int-problem} from \eqref{eq:call-gauss}, which should
decrease the variance of the integrand.  Figure \ref{fig:vol1+3semc}
corroborates that the smoothing has the expected effect on the Monte Carlo
quadrature, but the effect seems to diminish in higher dimensions.  In case of
the QMC quadrature, the smoothing does not effect the convergence rate but it
does improve the convergence constant as well. Moreover, the effect is even
stronger as for the Monte Carlo quadrature.  The convergence constant of the
QMC quadrature relies on the variation of the integrand and hence we suspect a
larger decrease in the variation of the integrand than in the variance. This
may be explained by the fact that the variation of a function can be
calculated from the first mixed derivatives.  Thus, the variation strongly
depends on the smoothness of the integrand, in particular this dependence is
stronger than for the variance of the integrand.

\subsection{Acceleration by using a sparse-grid interpolant as a control variate}

Another option to exploit the smoothness of the integrand is to combine a
(Q)MC quadrature with a sparse-grid approximation.  To that end, we construct
a sparse-grid interpolant on the {smoothened} 
integrand in \eqref{eq:basket-int-problem},
that is we use sparse-grid quadrature nodes as interpolation points and
employ this interpolant as a control variate. To explain the concept of a
control variate, let us consider the integration problem of a function
\(f\colon \R^d \to \R\) and an approximation \(g\colon \R^d \to \R\) on
\(f\). We assume that it is easy to calculate \(E(g)\isdef\int_{\R^d} g(x) \d
x\). Then, we rewrite the integral as
\begin{equation}
\label{eq:cont}
\int_{\R^d} f(x) \d x = \int_{\R^d} f(x)-g(x) \d x +E(g). 
\end{equation}
Instead of using a (Q)MC estimate of the integral on the left-hand side of
\eqref{eq:cont}, we estimate the integral on the right-hand side. This means that 
the function \(g\colon \R^d \to \R\) serves as a control variate, see for example 
\cite{G04} for a more detailed description. Of course, the quality of the
control variate depends on how much the variance or the variation of \(f-g\)
is reduced compared with the variance or variation of \(f\).  Hence, it is
closely connected to the approximation quality of \(g\) on \(f\). 

{In our examples, we use a total degree sparse-grid interpolant as a control 
variate. To describe that in more detail, let us denote by 
\(\mathcal{I}_j\colon C(\mathbb{R})\to \mathcal{P}_{N_j}\) the interpolation 
operator at the \(N_j\) quadrature points of the Gau{\ss}-Hermite quadrature with 
\(N_j\) points. Then our sparse-grid interpolant \(g\) is given by 
\[
g = \sum_{\alpha\in \N_0^d: \|\alpha\|_{1} \le 2} \bigotimes_{j=1}^d(\mathcal{I}_{\alpha_j}-\mathcal{I}_{\alpha_j-1}) f
\]
with the convention \(\mathcal{I}_{-1}\equiv 0\) and the numbers of quadrature 
points \(N_0=1\), \(N_1=3\) and \(N_2=5\).}

\begin{remark}
  The evaluation of this sparse-grid interpolant at the (Q)MC quadrature
  points becomes quite costly, especially in high dimensions. Most likely,
  more efficient control variates could be used, for example by including only
  the five most important dimensions in the sparse-grid interpolant.
  Nevertheless, the aim here is to demonstrate that it is possible, due to the
  smoothing, to significantly improve the convergence behaviour of the (Q)MC
  quadrature by a sparse-grid control variate on a relatively low level but we
  do not incorporate an efficiency analysis in terms of computational times
  here.
\end{remark}

\begin{figure}[htb]
\includegraphics[width=.48\textwidth]{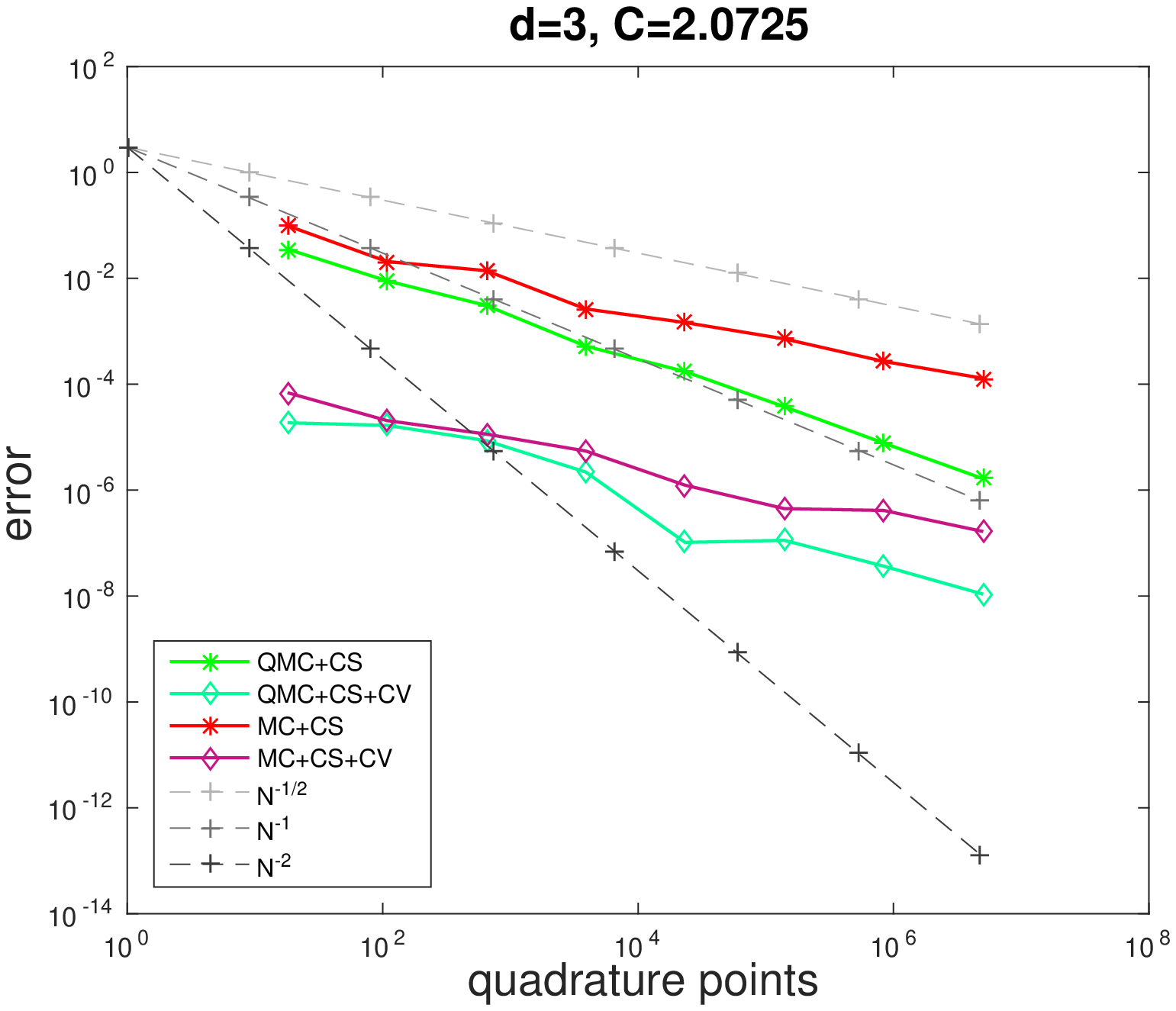}
\includegraphics[width=.48\textwidth]{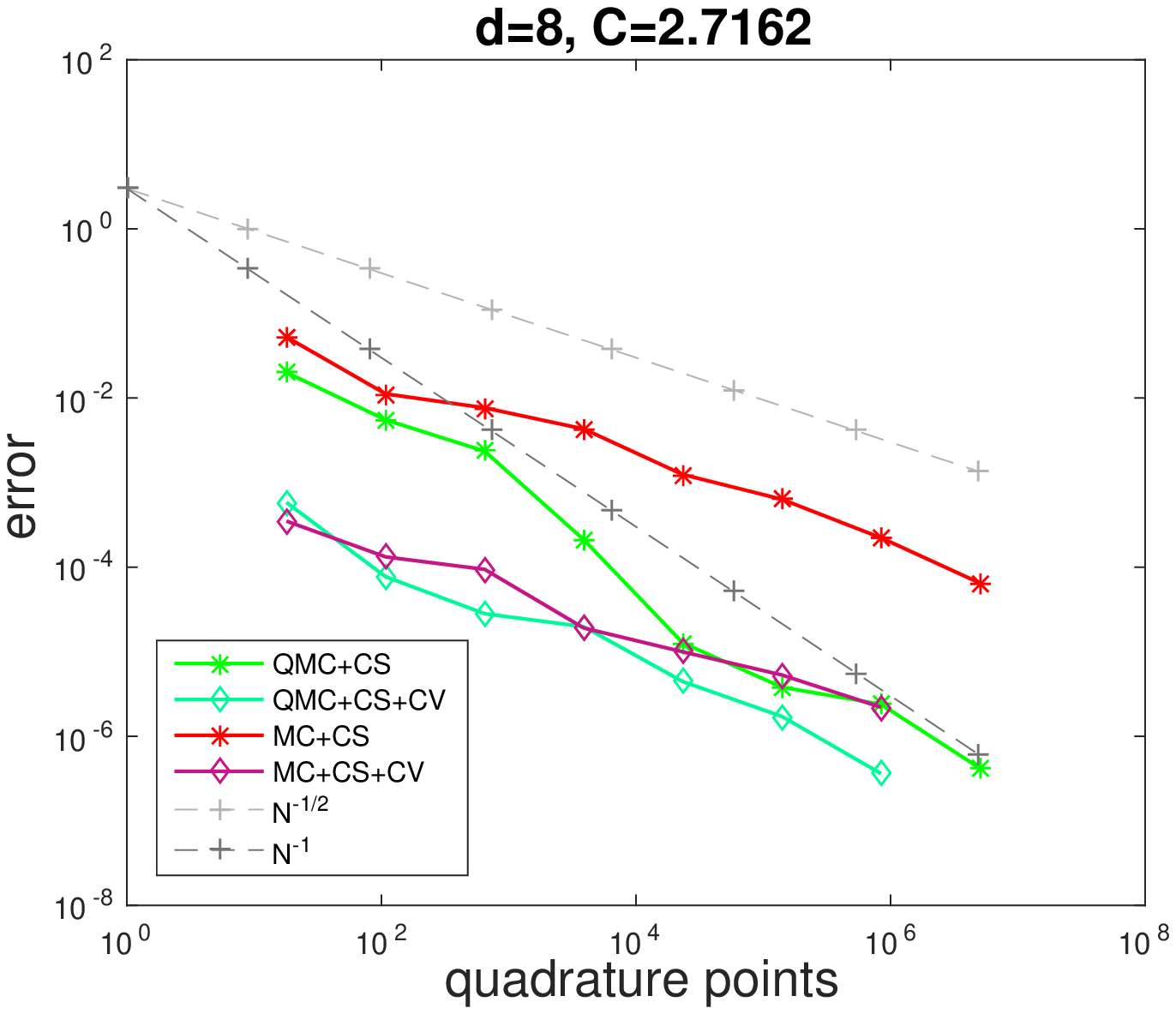}
\includegraphics[width=.48\textwidth]{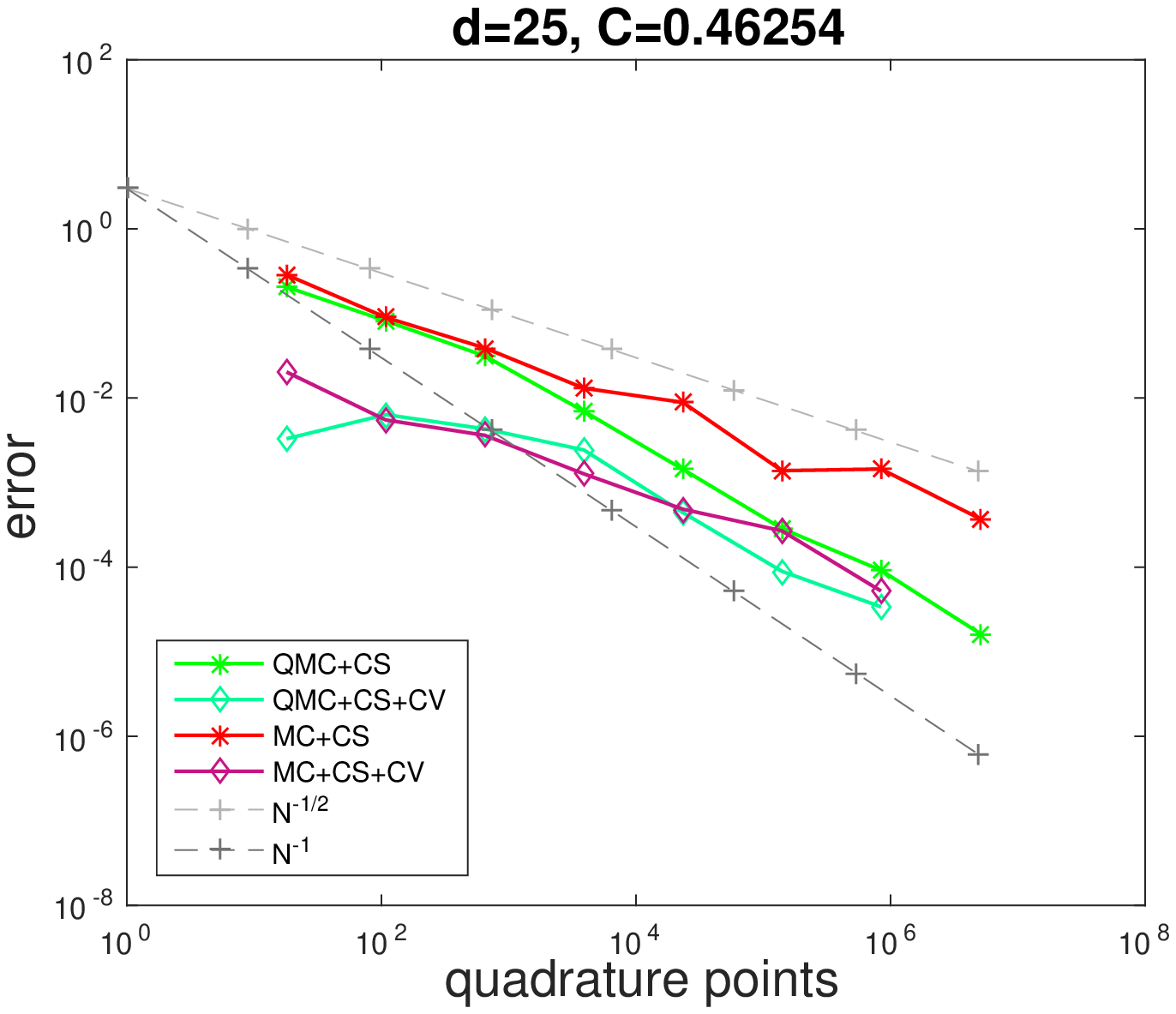}
\caption{\label{fig:vol1+3cont} Acceleration of the (Q)MC quadrature with a
  sparse-grid control variate for \(d=3\), \(d=8\) and \(d=25\) with
  volatilities selected randomly from the interval \([0.3,0.4]\).}
\end{figure}

The results of employing such a function as a control variate to improve the
convergence of the (Q)MC+CS quadrature are visualized in
Figure~\ref{fig:vol1+3cont}.  The error reduction is quite impressive for both
methods. In particular, the error of the MC+CS quadrature is reduced by a factor
of approximately \(10^3\) in \(d=3\) dimension and still by a factor \(10^2\)
in \(d=8\) and by a factor \(30\) in \(d=25\) dimensions while the convergence rate is preserved.
In the case of the QMC+CS quadrature, the constant is reduced by a similar factor
as in the Monte Carlo case. Although, the convergence rate seems to be slightly worse
in comparison with the QMC+CS quadrature, the quasi-Monte
Carlo quadrature with a sparse-grid control variate achieves the best error
behaviour of the four considered methods in Figure~\ref{fig:vol1+3cont} at least 
for \(d=3\) and \(d=8\).

\section{Numerical example 2: Multivariate Variance-Gamma setting}
\label{sec:numerical-example-2}

In our second numerical example, we consider the pricing of a basket option in
a multi\-variate Variance-Gamma model as introduced in \cite{LS06}. Therefore,
we recall that the multivariate extension of the univariate asset price
process \eqref{eq:var-gamma-price} is described as follows (cf.~\cite{LS15}
and Section~\ref{sec:introduction} above):
\begin{equation}
\label{eq:multvargam}
S_t^i = S_0^i\exp\Big((r+\omega_i)t+\theta_i \gamma_t+\sigma_i W_i(\gamma_t) \Big) 
\end{equation}
with 
\[
\omega_i = \frac{1}{\nu} \log\left(1-\frac 1 2 \sigma_i^2\nu -\theta_i\nu\right).
\]
We also incorporate here the deterministic interest rate \(r\) in order to
compare our results with those from \cite{LS15}.  The correlated
\(d\)-dimensional Brownian motion \(W\) in \eqref{eq:multvargam} is as in
\eqref{eq:GBM} given by its correlation matrix
\(\rho=\left(\rho_{i,j}\right)_{i,j=1}^d\) and its volatility vector
\(\sigma=[\sigma_{1},\ldots, \sigma_d]^\top\). The Gamma process \(\gamma_t\)
is independent from \(W\) and described by the parameter \(\nu\) via its
density function
\[
f_{\gamma_t}(y) = \frac{y^{1/\nu-1}}{\nu^{t/\nu}\Gamma(t/\nu)} e^{-y/\nu}.
\]
The calculation of a European basket call option at time \(T\) 
under the Variance-Gamma model leads then to 
\begin{equation}
\label{eq:baskvargam}
C_{\mathcal{B}} \coloneqq \int_{0}^{\infty}e^{-rT}E\left[ \left( \sum_{i=1}^d c_i S^i_T - K\right)^+
  \Bigg|\gamma_T=y\right]f_{\gamma_T}(y)\d y.
\end{equation}
Herein, the integrand is for every fixed \(y\ge 0\) just the value of a basket
call option according to \eqref{eq:basket-option}. Let us define
\begin{equation}\label{eq:setvargam}
\begin{aligned}
  w_i &= c_i S^i_0 e^{(r+\omega_i) T}, \quad i=1, \ldots, d, \\
  \Sigma_{i,j} &= \sigma_i \sigma_j \rho_{i,j} T, \quad i,j = 1, \ldots, d.
\end{aligned}
\end{equation}
Then, we can as in \eqref{eq:call-gauss} rewrite the integrand in terms of a
\(d\)-dimensional Gaussian vector $X^{y} = (X_1^{y}, \ldots, X_d^{y}) \sim
\mathcal{N}(0, y\cdot\Sigma)$ to
\[
E\left[ \left( \sum_{i=1}^d c_i S^i_T - K\right)^+
 \Bigg|\gamma_T=y\right]=E\left[ \left( \sum_{i=1}^d e^{\theta_i y}w_i e^{X_i} - K\right)^+
 \Bigg|\gamma_T=y\right]. 
\]
Hence, we can apply the technique from Section \ref{sec:smoothing-payoff} to
equation \eqref{eq:baskvargam}. Therefore, we recall the decomposition of the
matrix \(\Sigma=VDV^{\top}\) according to Lemma \ref{lem:rank-reduction}. The
first row of the matrix \(V\) is the vector \(v=[1,\ldots,1]^\top\) and we
denote the entries of the diagonal matrix by
\(D=\operatorname{diag}(\lambda_1^2,\ldots,\lambda_d^2)\).  Continuing in the
same fashion as in Section \ref{sec:smoothing-payoff}, we end up with the
equivalent integration problem (cf.~\eqref{eq:basket-int-problem}),
\begin{equation}
\label{eq:smoothvargam}
\begin{aligned}
C_{\mathcal{B}}&= \int_{0}^{\infty}e^{-rT}
E\left[ C_{BS}\left(h_y\left( \sqrt{y\barD} Z \right)
      e^{y\lambda_1^2/2}, K, \sqrt{y}\lambda_1 \right) \right]
      f_{\gamma_T}(y)\d y,\\  
Z &\sim \mathcal{N}\left(0, I_{d-1} \right), \  \sqrt{\barD} =
  \diag(\lambda_2, \ldots, \lambda_d).
  \end{aligned}
\end{equation}
Herein, the function \(h_y\) is given similar as in \eqref{eq:h-def} by 
\[
  h_y(z_2, \ldots, z_d) \coloneqq \sum_{i=1}^d e^{\theta_i y}w_i 
  \exp\left(\sum_{j=2}^d V_{i,j} z_j \right), \quad \barz = (z_2, \ldots, z_d) \in \R^{d-1}.
\]
Note that the integrand in \eqref{eq:smoothvargam} is very easy to calculate
with respect to \(y\) since we only need to incorporate the factor
\(e^{\theta_i y}\) in front of each weight \(w_i\) and scale the matrix \(D\)
by \(y\). Thus, the decomposition of the correlation matrix in view of Lemma
\ref{lem:rank-reduction} only has to be computed once although the correlation
matrix of the Gaussian vector \(X^{y}\) depends on the parameter \(y\).

\begin{figure}[htb]
\includegraphics[width=.48\textwidth]{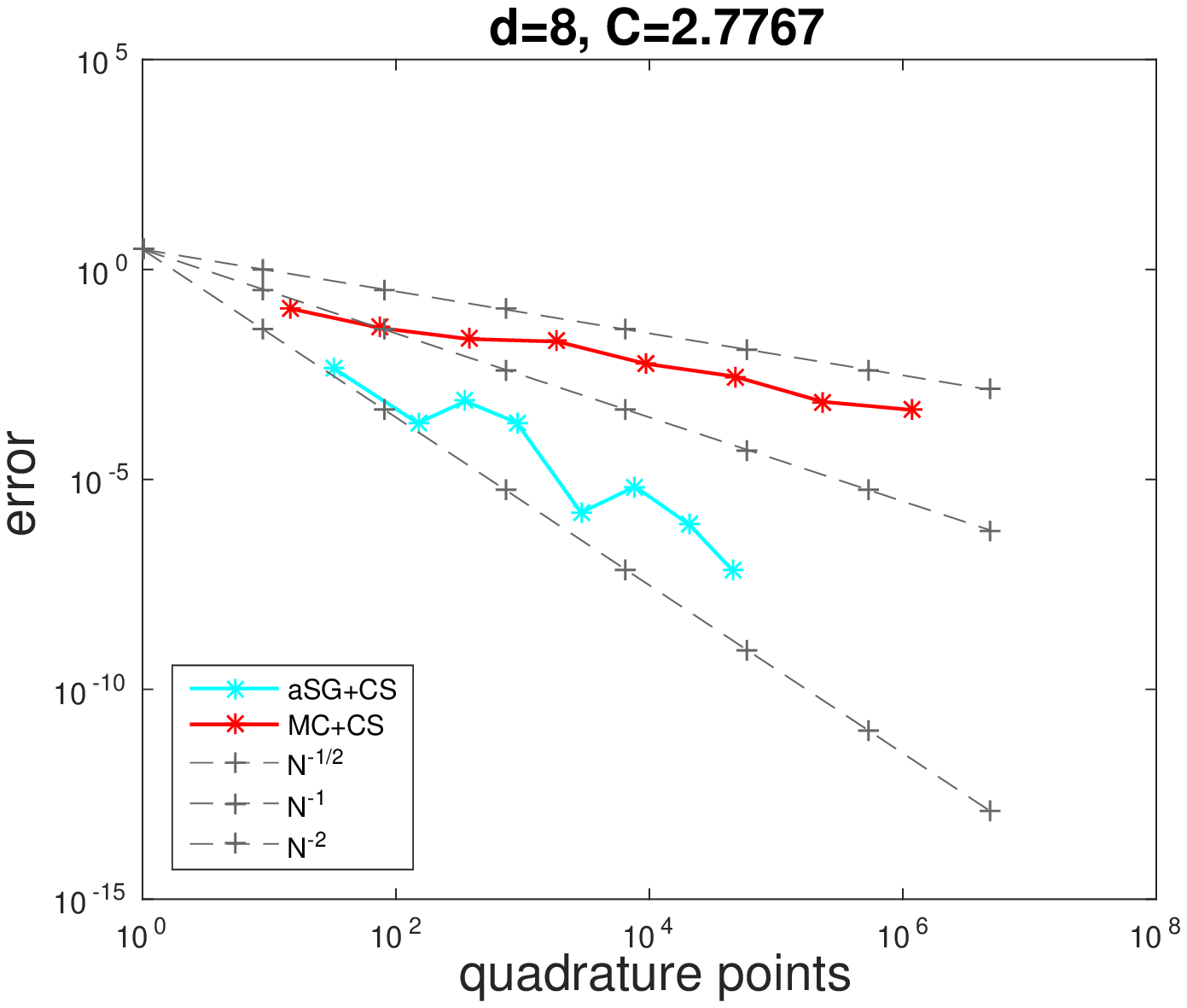}
\includegraphics[width=.48\textwidth]{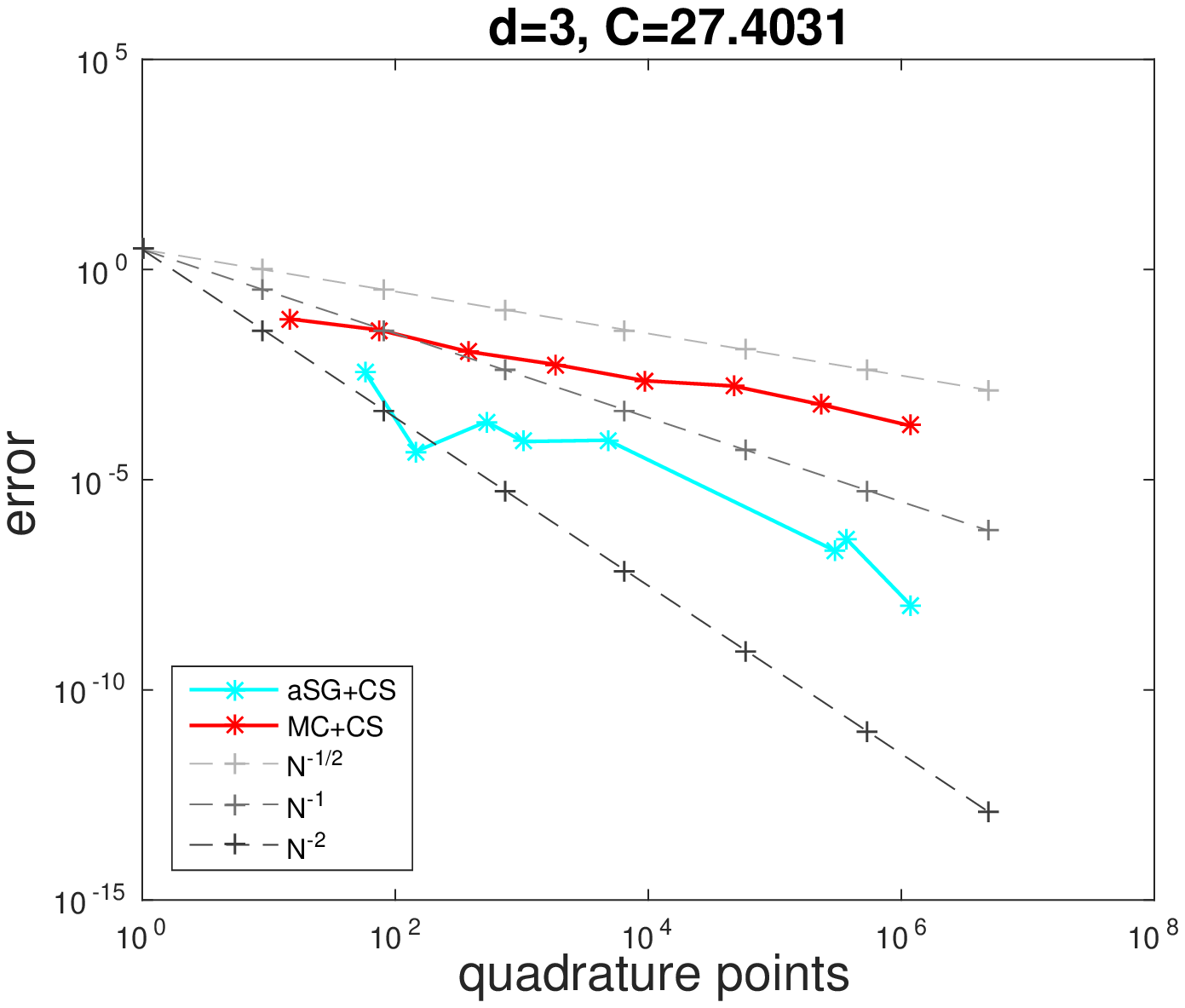}
\caption{\label{fig:vargam} Errors of an ATM basket call under a
  Variance-Gamma model with parameters \(\nu=0.3\) and \(\theta_i\in
  [-0.1,0.05]\) for \(d=8\) assets on the left and for an example from
  \cite{LS15} with \(d=3\) assets on the right.}
\end{figure}

In Figure \ref{fig:vargam}, we present two examples for basket option pricing
under the Variance-Gamma model. The first picture on the left-hand side
depicts the error of the calculation of an ATM basket call
(cf.~\eqref{eq:baskvargam}). We choose the parameters \(r=0\) and \(\nu=0.3\)
in \eqref{eq:multvargam} deterministically and randomly select
\(\theta_i\in[-0.1,0.05]\). Moreover, the correlation matrix \(\rho\), the
volatilities \(\sigma_i\) and the initial values \(S_0^i\) in \(d=8\) dimensions 
are constructed as in Section \ref{sec:mult-black-schol}.  We compare the 
convergence of the MC quadrature and the adaptive sparse-grid quadrature 
for the \(d\)-dimensional integral in \eqref{eq:smoothvargam}.  Note that the
integration domain and the density function in \eqref{eq:smoothvargam} are
given by
\[
\Gamma = [0,\infty]\times \R^{d-1},\qquad p(y,z_2,\ldots,z_d) = 
f_{\gamma_T}(y) \cdot \frac{1}{(2\pi)^{d/2}}\exp\left(\half\sum_{i=2}^d z_i^2\right). 
\]
Hence, we use \(d\)-dimensional random vectors where the first component is
distributed with respect to \(f_{\gamma_T}\) and independent to the remaining
\(d-1\) variables which are normally distributed and independent as well, as
samples for the Monte-Carlo quadrature. In case of the adaptive sparse-grid
quadrature, we apply tensor products of difference quadratures rules
(cf.~\eqref{eq:gensg}), where we use differences of generalized Gauss-Laguerre
quadrature rules as the quadrature sequence in the first variable. In the
remaining variables, we set the univariate quadratures as in Section
\ref{sec:mult-black-schol}. Afterwards, we select the indices which are
included in the sparse-grid adaptively as described in Section
\ref{sec:sparsegridconstr}.  As expected, the MC method converges exactly with
a rate \(N^{-1/2}\).  Moreover, the result demonstrates that the adaptive
quadrature outperforms the MC method even in this Variance-Gamma example with
an observed rate of nearly \(N^{-2}\).

The second numerical example is taken from the recent work \cite{LS15} and
stems originally from a parameter fitting of the Variance-Gamma model in
\cite{MCC98}. It describes a 3D model as in
\eqref{eq:multvargam} where \(\theta=[-0.1368, -0.056, -0.1984]^\top\),
\(\sigma=[0.1099, 0.1677, 0.0365]^\top\) and \(S_0=[100, 200, 300]^\top\).
Additionally, the weight vector \(c=[1/3, 1/6, 1/9]^\top\) and the correlation
matrix
\[\rho = 
\begin{pmatrix}
1 &0.6 &0.9 \\
 0.6&1 & 0.8\\
0.9 &0.8 &1 
\end{pmatrix} 
\]
were used. In \cite{LS15}, several different settings for the parameter
\(\nu\) and the strike price \(K\) are considered. We restrict ourselves to
the setting \(\nu=0.5\) and \(K=75\), which corresponds to an ITM basket.  On
the right-hand side of Figure \ref{fig:vargam} the convergence results for the
MC and the adaptive approaches are shown.  We observe that the MC quadrature
converges as before. Although the convergence of the adaptive sparse-grid
quadrature is still better than that of the Monte Carlo method, an exponential
rate as could be expected for such a low-dimensional example cannot be
obtained.  This deterioration in the convergence rate does not depend on the
Variance-Gamma setting but, as mentioned earlier, there is a
connection of the smoothing to the entries of the diagonal matrix \(D\) from
Lemma \ref{lem:rank-reduction}. For the considered example, the matrix \(D\)
has the entries \(\lambda_1^2=0.00023\), \(\lambda_2^2=0.03432\) and
\(\lambda_3^2=0.00652\). In particular, the small value of \(\lambda_1^2\)
explains the relatively small smoothing effect. In view of \eqref{eq:CSwithv}, 
the vector \(v=\mathbf{1}\) in Lemma
\ref{lem:rank-reduction} can be replaced by any other vector \(0\neq v \in
\{0,1\}^d\) to obtain a closed-form expression in Lemma
\ref{lem:BS-formula}.
\begin{figure}[htb]
\includegraphics[width=.48\textwidth]{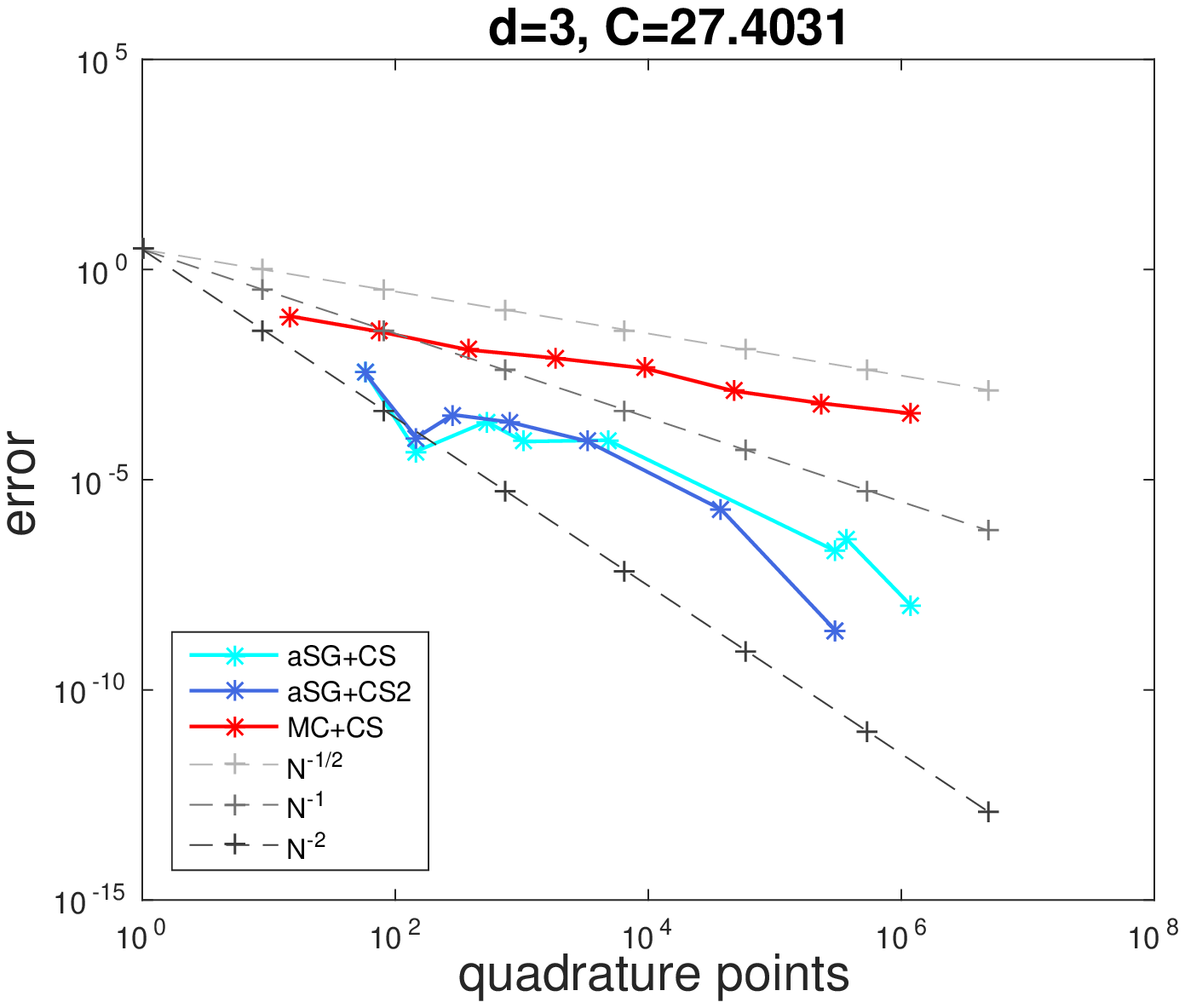}
\includegraphics[width=.48\textwidth]{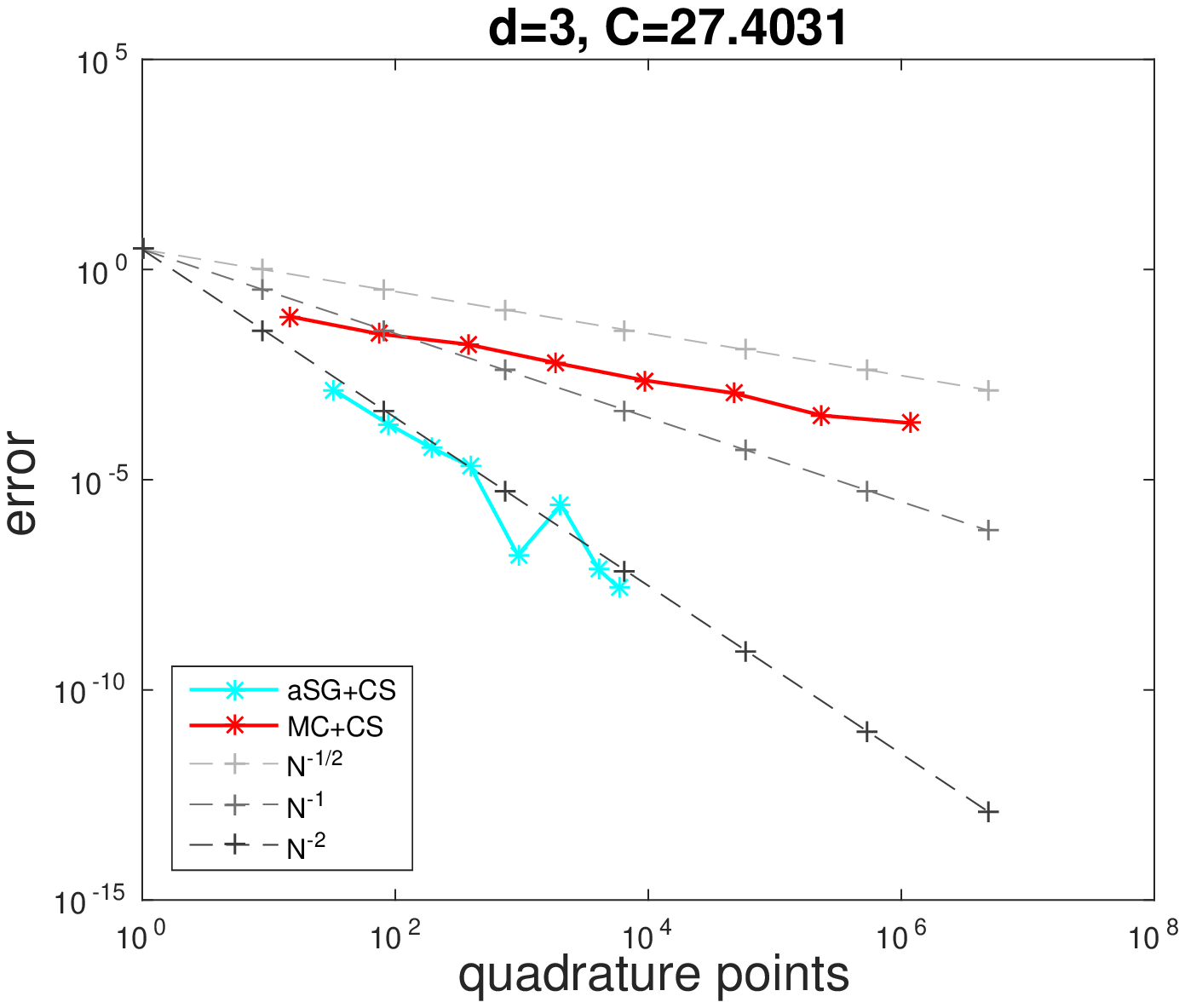}
\caption{\label{fig:vargam2} Errors for the example from \cite{LS15} with
  \(d=3\) assets. On the left-hand side, we included the convergence when the
  vector \(\mathbf{1}\) is replaced by \(v=[1,1,0]^\top\).  On the right-hand
  side, we used the modified volatility \(\sigma_3=0.1365\).}
\end{figure}
Therefore, we also investigated the convergence behaviour when we use a
vector \(v\neq \mathbf{1}\). {We tested all possible choices 
of vectors \(v\in\{0,1\}^3 \setminus \{0\}\) 
and observed that \(v=[1,1,0]^\top\) is best possible choice for this example. 
On the left-hand side of Figure \ref{fig:vargam2}, we compare the result from 
the right-hand side of Figure \ref{fig:vargam} with this best choice of \(v\) 
denoted by (aSG+CS2) and observe an improved convergence, which comes together 
with a size of \(\lambda_{1,v}^2=0.00109\), i.e.~\(\lambda_{1,v}^2\) 
is five times as high as \(\lambda_{1}^2\).} Nevertheless, $\lambda_{1,v}^2$ is still quite
small compared with \(\lambda_{2,v}^2=0.03294\) and, hence, the improvement in
the convergence is not that extraordinary. This leads to the supposition that
the considered example is not that well suited for our proposed method.  In
particular, the low value of \(\sigma_3=0.0365\) compared with the other
volatilities seems to have a negative effect on the smoothing.  Hence, we
tested this example also for the modified volatility \(\sigma_3=0.1365\).  The
results for this case also depicted on the right-hand side of figure
\ref{fig:vargam2} show a drastically improved convergence. Furthermore, the
entries of \(D\) are given by \(\lambda_1^2=0.01034\), \(\lambda_2^2=0.02255\)
and \(\lambda_3^2=0.00526\), which demonstrates the influence of the
differences in the volatilities on the size of \(\lambda_1^2\) and thus on
the smoothing.

\section*{Conclusions}
\label{sec:conclusions}

In the context of basket options, we show that the inherent smoothing property
of a Gaussian component of the underlying can be used to mollify the integrand
(payoff function) without introducing an additional bias. Having obtained a
smooth integrand, we can now directly apply (adaptive) sparse-grid methods. We
observe that these methods are highly efficient in low and moderately high
dimensions. For instance, the error can be improved by two orders of magnitude
in dimension $25$ compared to (Q)MC methods. In dimension $3$, we even obtain
exponential convergence. {While the actual benefit of the
  smoothing method is very much problem dependent, we observed good results
  for the adaptive sparse grid method for the smoothed integrand up to
  dimension $35$ in the examples we considered.}
We have also discussed improvements for MC and QMC
methods by introducing the smoothed payoff. In the Monte Carlo case, we do not
observe a significant improvement in the computational error, as the variance
reduction seems rather negligible. For QMC methods---Sobol numbers, to be more
precise---we do see considerable improvements in the constant. As expected,
the rate stays the same.

We note that the method employed in this work is not restricted to basket
options in a multivariate Black-Scholes or Variance-Gamma setting, but can be
generalized considerably. For instance, each step of an Euler discretization
of an SDE corresponds to a Gaussian mixture model. Hence, the conditional
expectation of the final integrand, given all the Brownian increments except for the
last one, is in the form of a Gaussian integral of the payoff function
w.r.t.~to a normal distribution with possibly complicated mean vector and
covariance matrix. If this integral can be computed explicitly, then we can
directly obtain mollification of the payoff without introducing a bias.

Even if the integral cannot be computed in closed form, there may be use cases
for employing numerical integration. For instance, in the basket option case,
a fast and highly accurate numerical integration of the 1D log-normal
integral, coupled with regression/interpolation (to avoid re-computation of
the one-dimensional integral for each new (sparse) gridpoint) could turn out
to be more efficient than a numerical integration technique applied to the
full problem.

Finally, note that there are also clear limitations of the technique. For
instance, consider a variation of the basket options studied in this work,
namely a best-of-call option. Here, the payoff is given by
\begin{equation*}
  \left( \max_{i=1, \ldots d} S^i_T - K \right)^+
\end{equation*}
for log-normally distributed, correlated variables $S^i_T$ (in the
Black-Scholes setting). Clearly, we can use Lemma~\ref{lem:rank-reduction} to
construct a common normal factor $Y$ and other factors $Y_1, \ldots, Y_d$ (all
jointly normal, $Y$ independent of the rest), such that $S^i_T =
e^{Y}e^{Y_i}$. Therefore, for the price of the best-of-call option, we obtain
\begin{equation*}
  E\left[ \left( \max_{i=1, \ldots d} S^i_T - K \right)^+ \right] = E\left[
    \left( e^Y \max_{i=1, \ldots d} e^{Y_i} - K \right)^+ \right].
\end{equation*}
Taking the conditional expectation, we obtain the Black-Scholes formula
applied at $\max_{i=1, \ldots d} e^{Y_i}$, which is still a non-smooth
payoff. The mollification can only remove one source of irregularity in this
case, not all of them.  Indeed, as currently presented in this work, the
conditional expectation step is most effective when the discontinuity surface
of the option's payoff has co-dimension one.

\medskip
\paragraph{Acknowledgements}
  R. Tempone is a member of the KAUST Strategic Research Initiative,
  Center for Uncertainty Quantification in Computational Sciences and
  Engineering. C. Bayer and M. Siebenmorgen received support for research
  visits related to this work from R. Tempone's KAUST baseline funds.

\bibliographystyle{plain}
\bibliography{smoothing}

\end{document}